\documentclass[runningheads]{llncs}
\usepackage{array}
\usepackage{booktabs}
\usepackage{multirow}
\usepackage{amsmath}
\usepackage[T1]{fontenc}
\usepackage{xcolor}
\usepackage{amsmath}
\usepackage{graphicx}
\usepackage{subcaption}
\usepackage{orcidlink}

\usepackage{graphicx}
\usepackage{mathtools}
\usepackage{amssymb}
\usepackage{cleveref}

\usepackage{amsmath, amsthm, amssymb}
\newif\ifextendedversion
\extendedversiontrue

\theoremstyle{definition} 
\newtheorem{assumption}{Assumption}[section]

\newcommand{\samethanks}[1][\value{footnote}]{\footnotemark[#1]}

\begin{document}

\title{Understanding Federated Learning Through the Lens of Mechanism Design: The Role of Data Heterogeneity}

\titlerunning{Understanding Federated Learning through Mechanism Design}

\author{Lina Alkarmi\thanks{Equal contribution.}\inst{1}\orcidlink{0000-0003-3097-4781} \and
Po-Yen Chen\samethanks\inst{1}\orcidlink{0009-0003-6079-1182} \and
Mingyan Liu\inst{1}\orcidlink{0000-0003-3295-9200}}

\institute{University of Michigan, Ann Arbor, MI 48109, USA \\
\email{\{lalkarmi, cpoyen, mingyan\}@umich.edu}}
\authorrunning{L. Alkarmi et al.}

\maketitle        
\begin{abstract}
Federated learning (FL) requires effective incentive mechanisms to motivate data sharing and prevent strategic free-riding. Recent FL mechanisms such as the Shapley value mechanism $\mathcal{M}^\text{Shap}$ guarantee reciprocal fairness for agents. However, a complete analysis of how such mechanisms impact social optimality and individual rationality under realistic, standalone outside options remains unknown. In this paper, we address this gap by adapting the classical Externality mechanism $\mathcal{M}^E$ to the federated learning setting. We conduct a comparison of $\mathcal{M}^\text{Shap}$ and $\mathcal{M}^E$ across three dimensions: social optimality, individual rationality, and fairness/reciprocity. First, we establish that $\mathcal{M}^\text{Shap}$ generally does not maximize social welfare because its marginal incentives drive agents to over-contribute resources, while $\mathcal{M}^E$ maximizes social welfare by design. Second, we evaluate participation incentives through the individual rationality gap when considering agents' outside options as standalone training on their own data. We find that both mechanisms ensure individual rationality in homogeneous settings. We further show that under mild conditions, $\mathcal{M}^E$ maintains this guarantee under agent heterogeneity, whereas $\mathcal{M}^\text{Shap}$ does not. Third, we demonstrate that while $\mathcal{M}^\text{Shap}$ maintains perfect reciprocity by design, $\mathcal{M}^E$ generally does not, and only ensures that individual benefits match Shapley contributions at symmetric equilibria under homogeneity, as it sacrifices individual fairness to maximize collective welfare under heterogeneity. Empirical simulations validate our theoretical findings and illustrate a tradeoff between reciprocal fairness and social efficiency.

\keywords{Federated learning, data privacy and security, mechanism design, Nash implementation}
\end{abstract}
\section{Introduction}

Federated learning (FL) is a collaborative machine learning framework in which a group of agents with private datasets jointly train a shared model via a central coordinating server, without directly sharing raw data \cite{pmlr-v54-mcmahan17a}, therefore preserving data privacy. A central challenge in FL is incentivizing agents to contribute meaningfully. Data sharing is costly, and strategic agents may free-ride on the contributions of others, leading to suboptimal learning outcomes. In fact, \cite{karimireddy2022mechanisms} show that giving agents unconditional access to the jointly trained model leads to free-riding where almost no agent contributes any data. Similarly, \cite{9746813} models FL as a public goods game and proves that a social dilemma exists in which selfish behavior at the Nash equilibrium leads to a loss of social welfare. Moreover, the training data may not be homogeneous across agents. In particular, agents may face heterogeneous privacy requirements, leading them to inject varying levels of noise into their contributions and resulting in heterogeneous data quality, a practically important source of agent heterogeneity in FL. These observations motivate the need for carefully designed incentive mechanisms.

The mechanism design perspective of FL addresses this challenge by designing payment schemes that align individual incentives with a collective goal. Here, a mechanism specifies a payment to each agent as a function of the collaboration profile of all agents. By receiving a payment that depends on the agents' contributions, each agent has a financial incentive to contribute data beyond what they would contribute without payment. A desirable mechanism should induce a simultaneous-move game whose Nash equilibria satisfy several well-studied
properties. It should be \emph{budget-balanced}, 
meaning the server operates at no profit or loss. 
It should be \emph{socially optimal}, meaning the 
implemented contribution profile maximizes total 
welfare across all agents. It should be 
\emph{individually rational}, meaning no agent is 
made worse off by participating than training alone. Beyond these commonly considered properties, it may also be desirable for an equilibrium to be {\em fair}, which generally stipulates how much each agent contributes relative to how much they each benefit. Notions of fairness are known to often come at the expense of social welfare \cite{Bertsimas_fairness,caragiannis2012efficiency}, and are less studied in the literature of using mechanism design to solve resource allocation problems. On the other hand, FL is {\em not} a typical resource allocation problem and fairness appears important: here agents contribute a tangible, measurable commodity (data), making it natural to ask not just whether agents are willing to participate, but whether the resulting distribution of benefits corresponds to what each agent put in.

The mechanism design perspective of FL is particularly useful because it unlocks a toolkit of standard mechanisms from economics and game theory, developed to solve analogous problems in public goods provision, resource allocation, and network design. A natural question then arises: which of these classical mechanisms are appropriate for FL, and what properties do they guarantee?

Recent work has made significant progress on this front. The mechanism $\mathcal{M}^\text{Shap}$, proposed by \cite{pmlr-v267-murhekar25a}, is a Shapley value-based mechanism that guarantees each agent's net benefit reflects their marginal contribution to the federation, a notion of {\em reciprocal fairness}, while maintaining budget balance and admitting Nash equilibria reachable by best-response dynamics. Despite these attractive properties, two fundamental questions remain open. First, does $\mathcal{M}^\text{Shap}$ implement a socially optimal contribution profile, the one that maximizes total welfare across agents? Second, does it guarantee that every agent is better off participating than training alone?

In this paper, we investigate these questions by adapting a classical mechanism from public goods theory, the Externality mechanism (which we will refer to as $\mathcal{M}^E$) of \cite{hurwicz1979outcome}, to the FL setting, and conducting an in-depth comparison with $\mathcal{M}^\text{Shap}$ across three dimensions: social optimality, individual rationality, and reciprocity. $\mathcal{M}^E$ is a natural candidate for FL because it is designed to internalize the externalities of data sharing, aligning each agent's private incentives with the socially optimal outcome. We note that unlike \cite{9746813}, which models the global model as a {\em non-excludable} public good, we model FL as an {\em excludable} good where non-participating agents do not benefit from others' data contributions. By applying the Nash implementation framework of \cite{sharma2012localpublicgoodprovisioning}, we establish that $\mathcal{M}^E$ admits a game form in which the socially optimal contribution profile emerges as a Nash equilibrium (NE). 
Our main contributions are as follows.

\begin{enumerate}
    \item \textbf{Adaptation of the Externality mechanism to FL.} We adapt the Externality mechanism to the FL setting with heterogeneous data quality weights, deriving the payment rule and establishing Nash implementation via the framework of \cite{sharma2012localpublicgoodprovisioning}. To our knowledge, this is the first application of this classical mechanism to FL.

    \item \textbf{Social optimality.} We show that by construction, $\mathcal{M}^E$ implements the social optimum at every NE, while $\mathcal{M}^{\text{Shap}}$ cannot implement the social optimum in general. This follows from a fundamental incompatibility between the Shapley marginal incentive and the social marginal incentive, which we establish formally.

    \item \textbf{Individual rationality.} We introduce the Individual Rationality (IR) gap $\Delta_i(\mathcal{M}) = u_i(\vec{s}^*) - v^*_i$ as a quantitative measure of each agent's participation incentive, where $v^*_i$ is the utility agent $i$ could achieve by training alone (outside the mechanism). We show that under mild conditions, in the homogeneous setting where agents have similar types of data and cost, both mechanisms guarantee a non-negative gap, meaning both are able to incentivize agents to participate and share data. Under data heterogeneity, a difference between the two mechanisms emerges: $\mathcal{M}^E$ maintains a non-negative IR gap as long as agents' data does not harm one another, while $\mathcal{M}^{\text{Shap}}$ can suffer from negative IR gaps, meaning it can fail to incentivize an agent's participation even when its data benefits others.

    \item \textbf{Reciprocity.} We show that while $\mathcal{M}^{\text{Shap}}$ achieves perfect reciprocity by design across all settings, $\mathcal{M}^E$ generally does not. Instead, $\mathcal{M}^E$ ensures that each agent's total benefit exactly matches their Shapley contribution specifically at the symmetric Nash equilibria of the homogeneous case.
    
\end{enumerate}

Together, these results demonstrate the tradeoff 
between reciprocal fairness and social efficiency in 
FL mechanism design, and provide practitioners with 
a basis for choosing between the two 
mechanisms depending on their goals. The rest of this paper is organized as follows. Section \ref{sec:formulation} defines the problem and mechanism design perspective of FL. Section \ref{sec:mechanisms} introduces the two mechanisms we will compare. Sections \ref{sec:SO}, \ref{sec:IR}, \ref{sec:recip} analyze social optimality, individual rationality, and reciprocity, respectively. Section \ref{sec:sim} provides simulation results, and Section \ref{sec:discussion} analyzes the case where agents' data could cause harm to other agents. Finally, Section \ref{sec:conclusion} concludes the paper.
\section{Problem Formulation}
\label{sec:formulation}

We follow the problem formulation of \cite{karimireddy2022mechanisms,pmlr-v267-murhekar25a}. Consider a system of $n$ agents, indexed by the set $N = [n]$, collaborating to train a shared machine learning model via a central coordinating server. Each agent $i \in N$ contributes $s_i \in S_i = [0, \tau_i]$ data samples, and we denote the joint contribution vector by $\vec{s} = (s_1, \dots, s_n) \in \mathcal{S} := \bigtimes_j S_j$. 

Each agent $i$ receives a \emph{payoff} $a_i(\vec{s}) \in \mathbb{R}_{\geq0}$ (which quantifies the performance or accuracy of the federated model on player $i$'s local task) and incurs a \emph{cost} $c_i(s_i) \in \mathbb{R}_{\geq0}$ associated with contributing data. Without a mechanism, an agent's base utility is simply $v_i(\vec{s}) = a_i(\vec{s}) - c_i(s_i)$. Following \cite{pmlr-v267-murhekar25a}, we assume throughout that each payoff function $a_i$ is bounded, non-decreasing, and concave in $\vec{s}$. A \emph{payment scheme} $p$ assigns each agent $i$ a payment $p_i(\vec{s}) \in \mathbb{R}$ at a contribution vector $\vec{s}$, resulting in the following utility for agent $i$:

\begin{equation}
    u_i(\vec{s}) = a_i(\vec{s}) - c_i(s_i) + p_i(\vec{s}).
    \label{eq:utility}
\end{equation}

Agents are utility maximizers and strategically choose their contributions. 
For a payment scheme, a \emph{Nash equilibrium} (NE) is a contribution vector $\vec{s}^* \in \mathcal{S}$ such that no agent can improve their utility by unilaterally deviating \cite{mas1995microeconomic,Nisan2007algorithmic}. 

\begin{equation}
    u_i(\vec{s}^*) \geq u_i(s_i', \vec{s}_{-i}^*) \quad 
    \forall i \in N, \quad \forall s_i' \in S_i.
    \label{eq:NE}
\end{equation}

\subsection{Payoff and Cost}  
To capture the asymmetric data relationships that arise in practical federated learning settings, we consider a heterogeneous payoff model where agents' data contributions have varying relevance to one another: 

\begin{equation}
    a_i(\vec{s}) = 1 - (\sum_{j \in N} w_{ij} s_j + 1)^{-\beta},
    \label{eq:payoff}
\end{equation}
where $\beta \in (0, 1]$ controls the rate of diminishing returns on data, and $w_{ij}$ is a relevance weight capturing how useful agent $j$'s data is to agent $i$. We define the {\em weighted contribution} received by agent $i$, including its own, as $ \xi_i(\vec{s}) = \sum_{j \in N} w_{ij} s_j + 1$. The $+1$ in $\xi_i(\vec{s})$ normalizes the model so that $a_i(\vec{0}) = 0$, ensuring zero payoff when no data is contributed by any agent, while keeping $a_i(\vec{s})$ defined and bounded in $[0,1)$ for all $\vec{s} \geq 0$. 

\begin{assumption}
\label{assump:weights}
For all $i, j \in N$, $w_{ij} \geq 0$, with $w_{ii} = 1$. Moreover, for
every $j \in N$, at least three agents derive positive value from agent
$j$'s data, i.e. $|\{i \in N : w_{ij} > 0\}| \geq 3$. 
\end{assumption}

The non-negativity of $w_{ij}$ reflects that no  agent's data harms another's model, while $w_{ii}=1$  normalizes each agent's self-relevance. We  discuss in Section \ref{sec:discussion} what happens when this ``no-harm'' assumption is relaxed. The  requirement that at least three agents derive  positive value from each agent $j$'s data is needed  to apply $\mathcal{M}^E$ following the  framework of \cite{sharma2012localpublicgoodprovisioning}. For cases where $|\{i \in N : w_{ij} > 0\}| < 3$, the Nash implementation framework of \cite{sharma2012localpublicgoodprovisioning} can no longer be applied. Addressing such cases would likely require a different approach outside the Nash implementation, for instance a direct negotiation between the few beneficiaries involved, and we leave this extension to future work. 

Each agent incurs a linear cost $c_i(s_i) = \gamma_i s_i$, where $\gamma_i > 0$ is agent $i$'s marginal cost of contributing data. Note that when $w_{ij} = 1$ for all $i, j \in N$, this formulation specializes to the  setup in \cite{pmlr-v267-murhekar25a}, where every agent benefits equally from every unit of contributed data. The model presented above is thus a more general one.

\subsection{Desirable Properties}
\label{sec:desirable}

We focus on four properties: budget 
balance (BB), social optimality (SO), and individual rationality (IR)
are standard \cite{mas1995microeconomic,Nisan2007algorithmic}; reciprocity was introduced 
by \cite{pmlr-v267-murhekar25a}.
We define each in the context of the mechanism model outlined earlier and {\em every} NE $s^{*}$ of the induced game.

\begin{definition}[Budget Balance]
\label{def:BB}
A mechanism $\mathcal{M}$ is weakly budget-balanced (WBB) (resp. budget-balanced (BB)) if $\sum_{i \in N} p_i(\vec{s}^{*}) \leq 0$ (resp. $= 0$).
\end{definition}

\begin{definition}[Social Optimality]
\label{def:SO}
The total welfare at a contribution profile $\vec{s} \in 
\mathcal{S}$ is:
\begin{equation}
    W(\vec{s}) := \sum_{i \in N} \left[ a_i(\vec{s}) - c_i(s_i) 
    \right] ,
    \label{eq:welfare}
\end{equation}
which equals $\sum_{i \in N} u_i(\vec{s})$ under any BB mechanism. A contribution profile $\vec{s}^* \in \mathcal{S}$ is socially optimal if it maximizes total welfare:
\begin{equation}
    \vec{s}^* \in \arg\max_{\vec{s} \in \mathcal{S}} W(\vec{s}).
    \label{eq:SO_def}
\end{equation}
A mechanism $\mathcal{M}$ achieves social optimality (SO) if the contribution profile at every Nash equilibrium is socially optimal.
\end{definition}

\begin{definition}[Individual Rationality]
\label{def:IR}
Agent $i$'s maximum standalone utility $v_i^*$ is 
\begin{equation}
    v_i^* = \max_{s_i \in S_i} \left[ 
    a_i(s_i, \mathbf{0}_{-i}) - c_i(s_i) \right],
    \label{eq:standalone}
\end{equation}
where $\mathbf{0}_{-i}$ denotes the zero contribution 
vector for all agents other than $i$. A mechanism 
$\mathcal{M}$ is individually rational (IR) if every 
agent prefers participating over training alone at 
every NE:
\begin{equation}
    u_i(\vec{s}^*) \geq v_i^*,  \quad \forall \vec{s}^* 
    \in \mathrm{NE}(\mathcal{M}), \quad \forall i \in N.
\end{equation}
\end{definition}

\begin{remark}
    In FL, agents always retain the option to train independently on their own data, making $v_i^*$ the natural outside option. This refines the standard IR condition of \cite{pmlr-v267-murhekar25a}, which requires $u_i(\vec{s}^*) \geq 0$ for all $\vec{s}^* \in \mathrm{NE}(\mathcal{M})$ and all $i \in N$. Since $v_i^* \geq a_i(\vec{0}) - c_i(0) \geq 0$, our condition is at least as strong, and the two coincide when agents have no meaningful standalone option.
\end{remark}

The final property, reciprocity, compares each agent's benefit from the mechanism to its contribution, measured by the Shapley value. The \emph{Shapley value} of agent $i$ at contribution vector $\vec{s}$, denoted $\varphi_i^A(\vec{s})$, measures agent $i$'s average marginal contribution to the total payoff $A(\vec{s}) := \sum_{i \in N} a_i(\vec{s})$ across all possible coalition orderings \cite{pmlr-v267-murhekar25a,Nisan2007algorithmic}, and is defined as:

\begin{equation}
    \varphi_i^A(\vec{s}) = \frac{1}{n} \sum_{X \subseteq N \setminus \{i\}} 
    \binom{n-1}{|X|}^{-1} \left[ A(\vec{s}[X \cup \{i\}]) - 
    A(\vec{s}[X]) \right],
    \label{eq:shapley}
\end{equation}
where $\vec{s}[X]$ denotes the contribution vector $\vec{s}$ restricted to agents in $X$, i.e. $\vec{s}[X]_j = s_j$ for $j \in X$ and $\vec{s}[X]_j = 0$ for $j \notin X$. A well-known property of the Shapley value is that it exactly distributes total payoff among agents: $\sum_{i \in N} \varphi_i^A(\vec{s}) = A(\vec{s})$ for all $\vec{s} \in \mathcal{S}$. 

\begin{definition}[Reciprocity~{\cite{pmlr-v267-murhekar25a}}]
\label{def:reciprocity}
The reciprocity of a mechanism $\mathcal{M}$ is:
\begin{equation}
    \mathrm{Reciprocity}(\mathcal{M}) := \min_{\vec{s} \in 
    \mathrm{NE}(\mathcal{M})} \min_{i \in N} 
    \frac{a_i(\vec{s}) + p_i(\vec{s})}{\varphi_i^A(\vec{s})}.
    \label{eq:reciprocity}
\end{equation}
A mechanism is fully reciprocal if $\mathrm{Reciprocity}(\mathcal{M}) 
= 1$, meaning every agent's total benefit is at least their Shapley contribution at 
every Nash equilibrium. 

\end{definition}

For $\vec{s}^* \in \mathrm{NE}(\mathcal{M})$, we write $r(\vec{s}^*) := \min_{i \in N} \frac{a_i(\vec{s}^*) + p_i(\vec{s}^*)}{\varphi_i^A(\vec{s}^*)}$ for the reciprocity attained at $\vec{s}^*$, so that $\mathrm{Reciprocity}(\mathcal{M}) = \min_{\vec{s}^* \in \mathrm{NE}(\mathcal{M})} r(\vec{s}^*)$.
\section{The Two Mechanisms}
\label{sec:mechanisms}

In this section we introduce the two mechanisms we study and compare throughout the paper: $\mathcal{M}^{\text{Shap}}$, a Shapley value based mechanism from \cite{pmlr-v267-murhekar25a} that prioritizes reciprocal fairness, and the Externality mechanism $\mathcal{M}^{\text{E}}$, adapted from \cite{hurwicz1979outcome,TIFS_parinaz,sharma2012localpublicgoodprovisioning}, that prioritizes social efficiency. Both mechanisms are BB by construction. 

\subsection{The $\mathcal{M}^{\text{Shap}}$ Mechanism}
\label{sec:mshap}
$\mathcal{M}^{\text{Shap}}$ assigns 
each agent $i$ a payment equal to their Shapley share minus their 
direct payoff:
\begin{equation}
 p_i^{\text{Shap}}(\vec{s}) = \varphi_i^A(\vec{s}) 
    - a_i(\vec{s}). 
    \label{eq:mshap_payment}
\end{equation}
As a result, agent $i$'s utility is simply $u_i^{\text{Shap}}(\vec{s}) = \varphi_i^A(\vec{s}) - c_i(s_i)$. 
$\mathcal{M}^{\text{Shap}}$  is fully reciprocal by construction.
It operates as a direct mechanism: the central server announces the payment rule \eqref{eq:mshap_payment} to all agents, who then best respond by choosing their contributions strategically. \cite{pmlr-v267-murhekar25a} shows that the resulting Nash equilibrium can be obtained using best-response dynamics when $a_i(\vec{s})$ is concave in $\vec{s}$ and each cost function $c_i$ is non-decreasing and convex in $s_i$ for every agent $i \in N$.

\subsection{$\mathcal{M}^E$: The Externality Mechanism}
\label{sec:mech_ext}

Unlike $\mathcal{M}^{\text{Shap}}$, which is direct, $\mathcal{M}^E$ relies on message passing: each agent $i$ sends a message $m_i = (\vec{x}^i, \vec{\pi}^i$) consisting of a proposed contribution level for every agent together with a non-negative price for each agent's data. The outcome function of \cite{sharma2012localpublicgoodprovisioning} maps a message profile $\vec{m}$ to an implemented contribution profile and a payment to each agent, in such a way that each agent faces prices it cannot influence. At a NE $\vec{m}^*$ of the resulting message game, we write $\vec{s}^*$ for the implemented contribution profile and $p_i(\vec{s}^*)$ for agent $i$'s resulting payment, so that the properties defined in Section \ref{sec:desirable} apply to $\mathcal{M}^E$ exactly as they do in direct mechanisms, with $\text{NE}(\mathcal{M}^E)$ denoting the set of contribution profiles implemented at equilibria of the message game.

Under Assumption \ref{assump:weights}, each $u_i(\vec{s})$ is concave in $\vec{s}$, each $S_i = [0, \tau_i]$ is convex and compact, and every agent's data benefits at least three agents, so our FL setting constitutes a special case of the framework of \cite{sharma2012localpublicgoodprovisioning}. It follows from Theorems~1 and 2 of \cite{sharma2012localpublicgoodprovisioning} that Nash equilibria exist and that the contribution profile implemented at every NE of the induced game is socially optimal. Since $W$ is concave but not strictly concave in general, the welfare maximizer need not be unique, and different equilibria may implement different maximizers. Throughout our analysis of $\mathcal{M}^E$ we therefore fix an arbitrary NE and denote by $\vec{s}^*$ the socially optimal contribution profile it implements. Our results hold for every equilibrium satisfying Assumption \ref{assump:interior}:

\begin{assumption}
\label{assump:interior}
We restrict attention to Nash equilibria of $\mathcal{M}^E$ whose implemented contribution profile $\vec{s}^*$ is interior, i.e. $s_i^* \in (0, \tau_i)$ for all $i \in N$.
\end{assumption}

At a NE $\vec{s}^*$, the payment to 
each agent is determined by what is referred to as \emph{personalized 
price} $l^*_{ij}$, which is the marginal utility of agent $i$ with respect to agent $j$'s data contribution \ifextendedversion using $u_i = a_i(\vec{s}) - c_i(s_i)$ and Eq.~\eqref{eq:payoff}\fi. Under Assumption \ref{assump:interior}, the personalized
price has the following expression (\ifextendedversion the full derivation can be found in Appendix \ref{app:ext_deriv}\fi):
\begin{equation}
    l^*_{ij} = \begin{cases} 
    \beta w_{ij} \xi_i(\vec{s}^*)^{-\beta-1} & j \neq i \\[4pt]
    \beta \xi_i(\vec{s}^*)^{-\beta-1} - \gamma_i & j = i,
    \end{cases}
    \label{eq:prices}
\end{equation}

The payment to agent $i$ is then given by:
\begin{equation}
    p_i^E(\vec{s}^*) = -\sum_{j \in N} l^*_{ij} s^*_j 
    = \gamma_i s^*_i - \beta \xi_i(\vec{s}^*)^{-\beta-1}
    \left(\xi_i(\vec{s}^*) - 1\right),
    \label{eq:ext_payment_explicit}
\end{equation}
where $\gamma_i s^*_i$ reimburses agent $i$ for its cost; the second, tax term is proportional to the total benefit agent $i$ receives from the collaboration. Payments under $\mathcal{M}^E$ sum to zero at every profile by construction of the framework
\cite{sharma2012localpublicgoodprovisioning}.

\section{An Analysis of Social Optimality}
\label{sec:SO}

By construction of the game form presented in Section \ref{sec:mech_ext}, $\mathcal{M}^E$ implements the social optimum at every Nash equilibrium. We now show that $\mathcal{M}^\text{Shap}$ never does so at an interior equilibrium (the full proof can be found in \ifextendedversion Appendix \ref{app:mshap_SO}\fi).

\begin{theorem}
\label{thm:mshap_SO}
For any FL instance with payoffs \eqref{eq:payoff} and linear costs satisfying Assumption \ref{assump:weights}, no interior Nash equilibrium of $\mathcal{M}^{\text{Shap}}$ maximizes the total welfare $W$.
\end{theorem}

\begin{corollary}
\label{cor:overcontribution}
For any FL instance with payoffs \eqref{eq:payoff} and linear costs satisfying Assumption \ref{assump:weights}, let $\vec{s}^{NE}$ be any interior Nash equilibrium of $\mathcal{M}^{\text{Shap}}$. Then for every welfare maximizer $\vec{s}^{\circ}$, there exists an agent $i$ with $s_i^{NE} > s_i^{\circ}$. Moreover, if $w_{ij} = 1$ and $\gamma_i = \gamma$ for all $i, j \in N$, then $\|\vec{s}^{NE}\|_1 > \|\vec{s}^{\circ}\|_1$.
\end{corollary}

Social optimality asks each agent to contribute until the marginal value of its data falls to (equals) its marginal cost. Under $\mathcal{M}^{\text{Shap}}$, an agent instead contributes until its Shapley value falls to its marginal cost. The Shapley value averages the agent's marginal contribution over all sub-coalitions, including small ones holding little data, where the marginal value of data is higher owing to the diminishing return assumption. The Shapley margin therefore sits above the social margin, and the two conditions cannot hold simultaneously.  In essence, $\mathcal{M}^{\text{Shap}}$ rewards each agent for the value it would create in these hypothetical smaller federations, and this extra reward pushes contributions past the point where they benefit society. \Cref{cor:overcontribution} makes this precise: under $\mathcal{M}^{\text{Shap}}$, at any interior NE, some agent contributes strictly more than at any welfare maximizer, and in the homogeneous case, the total contribution strictly exceeds the latter.
\section{An Analysis of Individual Rationality}

\label{sec:IR}

Next, we compare $\mathcal{M}^\text{Shap}$ and $\mathcal{M}^E$ on individual rationality. We measure the \emph{IR gap} of agent $i$ at a NE $\vec{s}^*$ under mechanism $\mathcal{M}$:
\begin{equation}
    \Delta_i(\mathcal{M}) := u_i(\vec{s}^*) - v^*_i,
    \label{eq:IR_gap}
\end{equation}
where recall $v_i^{*}$ is agent $i$'s maximum utility when training alone. A positive (resp. negative) IR gap means the agent strictly prefers to participate (resp. opt-out of the federation). When the IR gap is negative, the mechanism can in principle provide agent $i$ an upfront payment of $e_i = \max(0, -\Delta_i(\mathcal{M}))$ to compensate for its shortfall, making participation at least as attractive as training alone. The IR gap therefore characterizes the minimum endowment needed to guarantee participation.  Our analysis below shows that while $\mathcal{M}^\text{Shap}$ requires data homogeneity to ensure a non-negative IR gap property, $\mathcal{M}^E$ provides this guarantee more generally.

\begin{assumption}
\label{assump:interior_solo}
We assume that each agent's standalone optimum is interior, i.e., $s_i^{\emph{solo}} \in 
(0, \tau_i)$ for all $i \in N$. 
\end{assumption}

\begin{theorem}
\label{thm:mshap_IR_homo}
In the homogeneous case ($w_{ij} = 1$ and $\gamma_i = \gamma$, 
$\forall i, j \in N$), 

\begin{equation}
    \Delta_i(\mathcal{M}^{\text{Shap}}) \geq 0 
    \quad \forall i \in N.
\end{equation}
\end{theorem}

\begin{theorem}
\label{thm:IR_ext}
Under Assumptions \ref{assump:weights}, \ref{assump:interior}, and
\ref{assump:interior_solo},
\begin{equation}
    \Delta_i(\mathcal{M}^E) \geq 0 \quad \forall i \in N.
\end{equation}
\end{theorem}
 
In essence, 
$\mathcal{M}^E$ and $\mathcal{M}^\text{Shap}$ differ under agent heterogeneity. $\mathcal{M}^E$ maintains a non-negative IR gap for any $w_{ij} \geq 0$, requiring no endowment regardless of data heterogeneity. In contrast, there exist heterogeneous instances where $\Delta_i(\mathcal{M}^{\text{Shap}}) < 0$ for some agent $i$. We demonstrate such instances and characterize the IR gap in Section \ref{sec:sim}. 

Unlike Theorem \ref{thm:IR_ext}, Theorem \ref{thm:mshap_IR_homo} does not require Assumption \ref{assump:interior_solo}: the argument compares each agent's equilibrium payoff to its standalone payoff directly, without relying on a marginal condition that could fail at a boundary solution. Intuitively, Theorem~\ref{thm:IR_ext} holds for $\mathcal{M}^E$ but not $\mathcal{M}^{\text{shap}}$ because under $\mathcal{M}^E$, an agent's utility depends only on the total weighted data they have access to, and the social optimum always gives each agent access to at least as much data as they would gather on their own.  In contrast, $\mathcal{M}^{\text{Shap}}$'s Shapley payments reward agents based on their {\em average contribution across all coalitions}; while this works when agents are data homogeneous, under heterogeneity, this averaging can over-charge agents whose data is narrowly useful relative to the benefit they receive, causing IR to fail.
\section{An Analysis of Reciprocity}

\label{sec:recip}

Finally, we compare $\mathcal{M}^\text{Shap}$ and $\mathcal{M}^E$ on reciprocity. $\mathcal{M}^{\text{Shap}}$ is fully reciprocal by design as mentioned earlier. The question is therefore how much reciprocity $\mathcal{M}^E$ retains while maintaining social optimality. 

\begin{theorem}
\label{thm:reciprocity_ME_homogeneous}
In the homogeneous case ($w_{ij} = 1$, $\gamma_i = \gamma$, and $\tau_i = \tau$ for all $i, j \in N$), under Assumptions \ref{assump:interior} and \ref{assump:interior_solo}, $\mathcal{M}^E$ admits a symmetric Nash equilibrium with $s_i^* = \bar{s} > 0$ for all $i \in N$. At every symmetric Nash equilibrium $\vec{s}^*$,
\begin{equation}
    a_i(\vec{s}^*) + p_i^E(\vec{s}^*) = \varphi_i^A(\vec{s}^*)
    \quad \forall i \in N,
\end{equation}
i.e. $r(\vec{s}^*) = 1$, the maximum attainable under budget balance.
\end{theorem}

At a symmetric equilibrium of a homogeneous instance, every agent contributes the same amount, benefits identically, and is owed exactly what it owes others, so the cost reimbursement and consumption tax in \eqref{eq:ext_payment_explicit} cancel and no payments are exchanged. With no payments, each agent's utility is simply their payoff, and identical agents making identical contributions means an agent's payoff equals its Shapley share; under $\mathcal{M}^\text{Shap}$, $p_i^{\text{Shap}}$ also vanishes at this profile, so the two payment rules converge. Furthermore, a BB mechanism can only redistribute, so reciprocating one agent above its Shapley share necessarily pushes another below; thus no equilibrium of a BB mechanism satisfies $r(\vec{s}^*) > 1$ \cite{pmlr-v267-murhekar25a}. Therefore, symmetric equilibria of $\mathcal{M}^E$ attain the maximum $r(\vec{s}^*)$ value possible for a BB mechanism.

Under heterogeneous weights or costs, payments no longer vanish and Shapley shares no longer coincide with individual payoffs, so $r(\vec{s}^*)$ can fall below $1$. We quantify this in Section \ref{sec:sim}. The result mirrors Section \ref{sec:SO} with the roles reversed. $\mathcal{M}^{\text{Shap}}$ is fully reciprocal but fails social optimality at every interior equilibrium, while $\mathcal{M}^E$ achieves social optimality at every equilibrium but attains $r(\vec{s}^*) = 1$ only at symmetric equilibria of homogeneous instances.

\section{Numerical Experiments and Simulation}
\label{sec:sim}
We perform two types of numerical evaluations: a model-based study generated directly from the analytical payoff model \eqref{eq:payoff},  and a simulation-based study in which agents participate in a full FL training loop on the MNIST dataset to verify that our findings hold under realistic training dynamics. In both studies, we assume all agents remain in their respective mechanisms regardless of their IR gap: relaxing this assumption and developing a distributed algorithm that dynamically converges to the $\mathcal{M}^E$ equilibrium remain directions for future work.

\subsection{Model-based Numerical Experiments}
\label{sec:theoretical}

We obtain the $\mathcal{M}^E$ solution via global optimization of $W(\vec{s})$ over the analytical payoff model \eqref{eq:payoff}, since its NE coincides with the social optimum. For $\mathcal{M}^{Shap}$, the equilibrium $s^*$ is obtained through best-response dynamics following the convergence result in \cite{pmlr-v267-murhekar25a}. Throughout, we use $n=3$ agents. 

\subsubsection{Social Optimality}
\label{subsec:homogeneous_theoretical}
We evaluate the social welfare and equilibrium data contributions in a homogeneous setting ($w_{ij}=1, \forall i,j$), varying the marginal cost parameter  $\gamma_i=\gamma \in \{0.02,0.015,0.01,0.007,0.005,0.003,0.001\}$ and the diminishing returns parameter $\beta \in \{0.005, 0.05, 0.5\}$. Recall that $\beta$ controls how steeply the payoff grows with data: larger $\beta$ means faster payoff growth, which drives stronger contribution incentives for both mechanisms.

As illustrated in Fig.~\ref{fig:cost_comparison}, the equilibrium contribution level $s^*$ under both mechanisms decreases monotonically as the marginal cost $\gamma$ increases. Across all values of $\gamma$ and $\beta$, $\mathcal{M}^{Shap}$ induces higher equilibrium contributions than $\mathcal{M}^E$, consistent with Corollary \ref{cor:overcontribution}. Importantly, the gap between the two mechanisms widens as $\beta$ increases: when payoffs grow faster with data, $\mathcal{M}^{\text{Shap}}$'s Shapley-based incentives push agents further past the socially optimal contribution level.

\begin{figure}[h]
    \centering
    \begin{subfigure}[b]{0.32\linewidth}
        \centering
        \includegraphics[width=\linewidth]{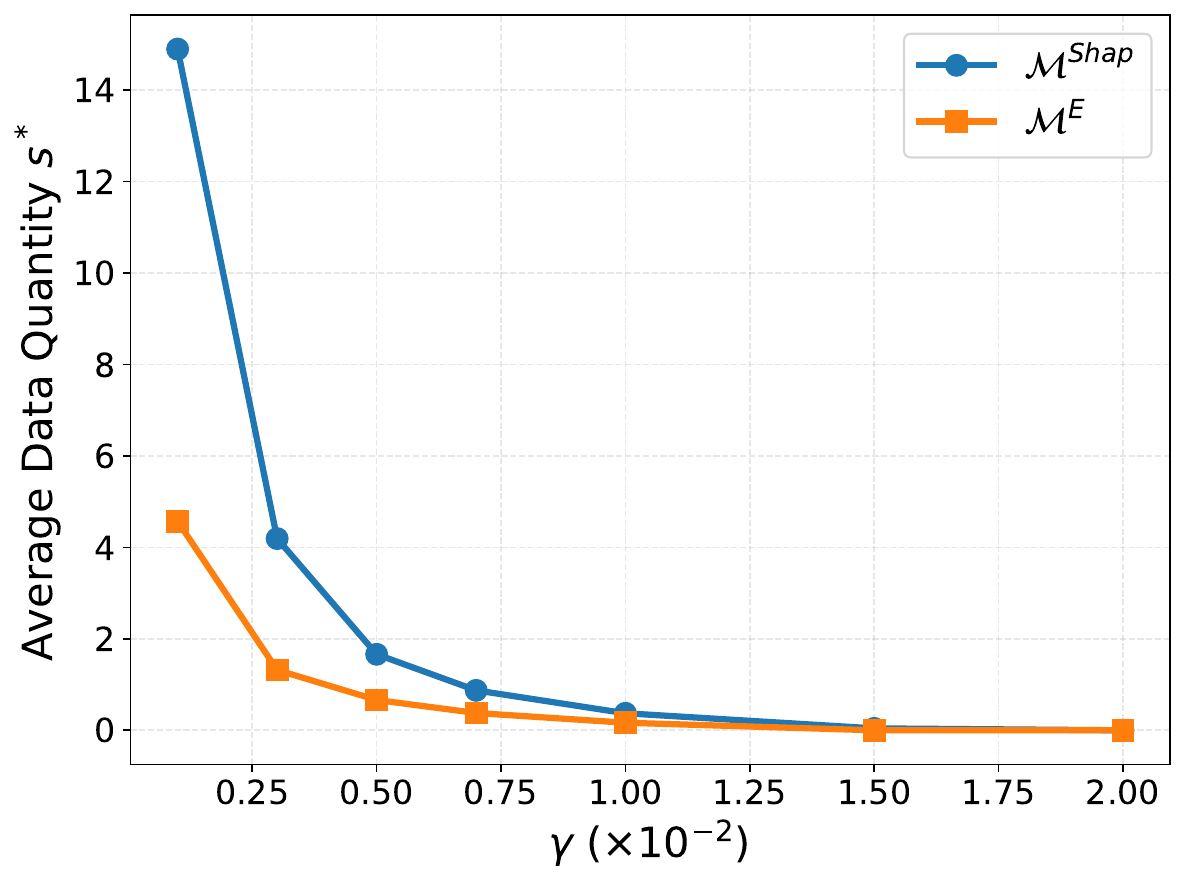}
        \caption{$\beta=0.005$}
        \label{fig:data_quantity_beta_0.005}
    \end{subfigure}
    \hfill
    \begin{subfigure}[b]{0.32\linewidth}
        \centering
        \includegraphics[width=\linewidth]{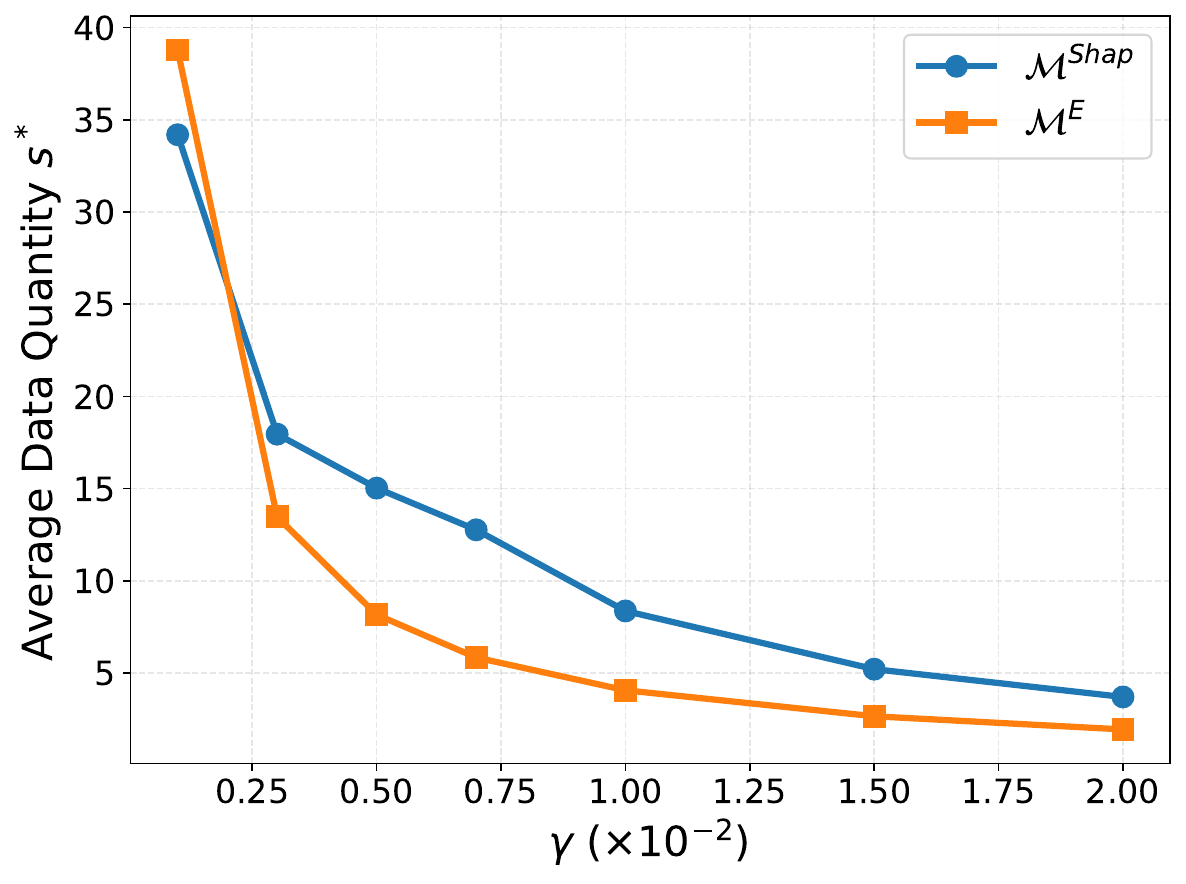}
        \caption{$\beta=0.05$}
        \label{fig:data_quantity_beta_0.05}
    \end{subfigure}
       \begin{subfigure}[b]{0.32\linewidth}
        \centering
        \includegraphics[width=\linewidth]{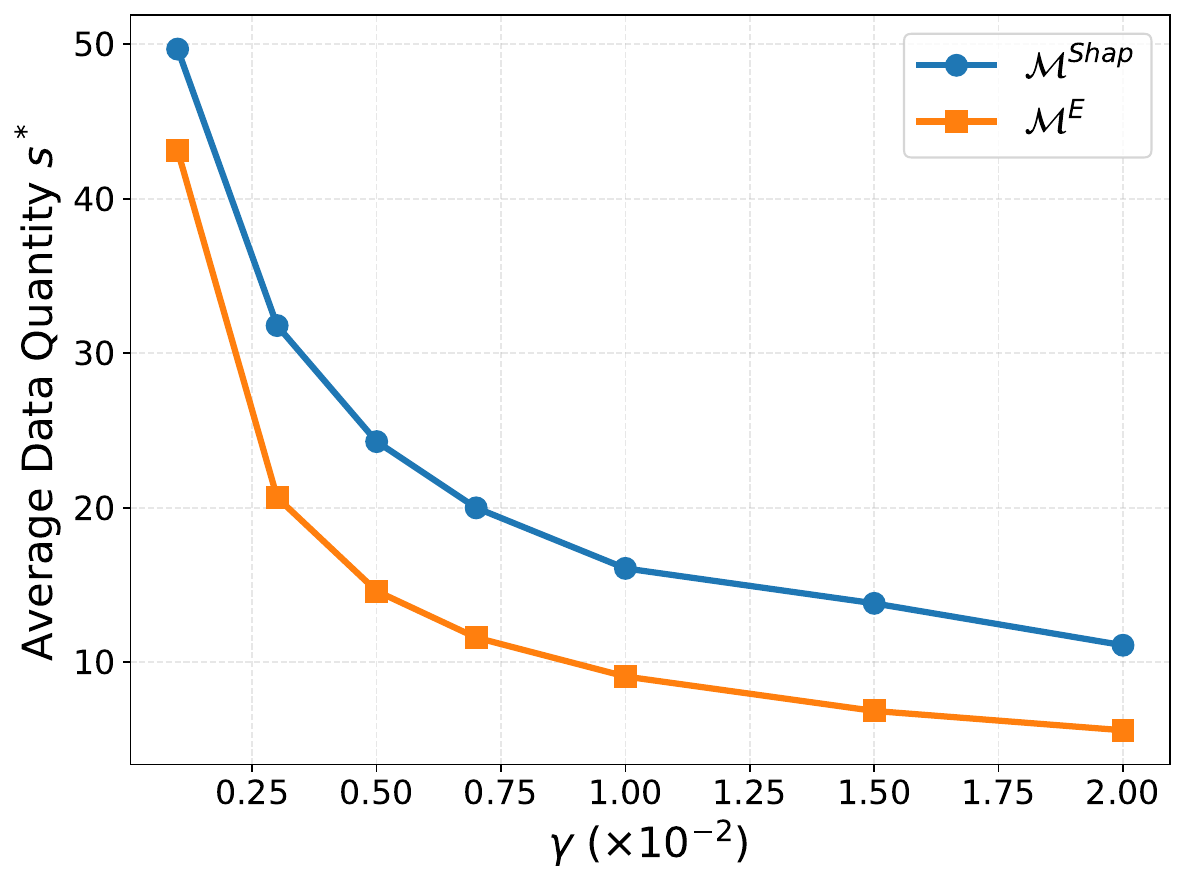}
        \caption{$\beta=0.5$}
        \label{fig:data_quantity_beta_0.5}
    \end{subfigure}
    \caption{Equilibrium data quantity $s^*$ under different $\beta, \gamma$ and $n=3$.} 
    
    \label{fig:cost_comparison}
\end{figure}

However, the additional contributions under $\mathcal{M}^{Shap}$ do not translate into higher social welfare, as seen in Fig.~\ref{fig:so}, which shows that despite contributing more data, $\mathcal{M}^{\text{Shap}}$ achieves strictly lower social welfare than $\mathcal{M}^E$ across all settings. This inefficiency arises because $\mathcal{M}^{Shap}$ tends to overestimate the value of data contributions, leading clients to contribute beyond the socially optimal level, as discussed earlier. This welfare gap is most pronounced at larger $\beta$ values, where the over-contribution under $\mathcal{M}^{\text{Shap}}$ is most severe. In contrast, $\mathcal{M}^E$ directly aligns incentives with social welfare and therefore avoids such over-contribution. At very small $\beta$ (Fig.~\ref{fig:so_beta_0.005}), where payoff growth is slow and contributions are low overall, this welfare gap shrinks, suggesting the two mechanisms behave more similarly when data has limited 
marginal value.

\begin{figure}[h]
    \centering
    \begin{subfigure}[b]{0.32\linewidth}
        \centering
        \includegraphics[width=\linewidth]{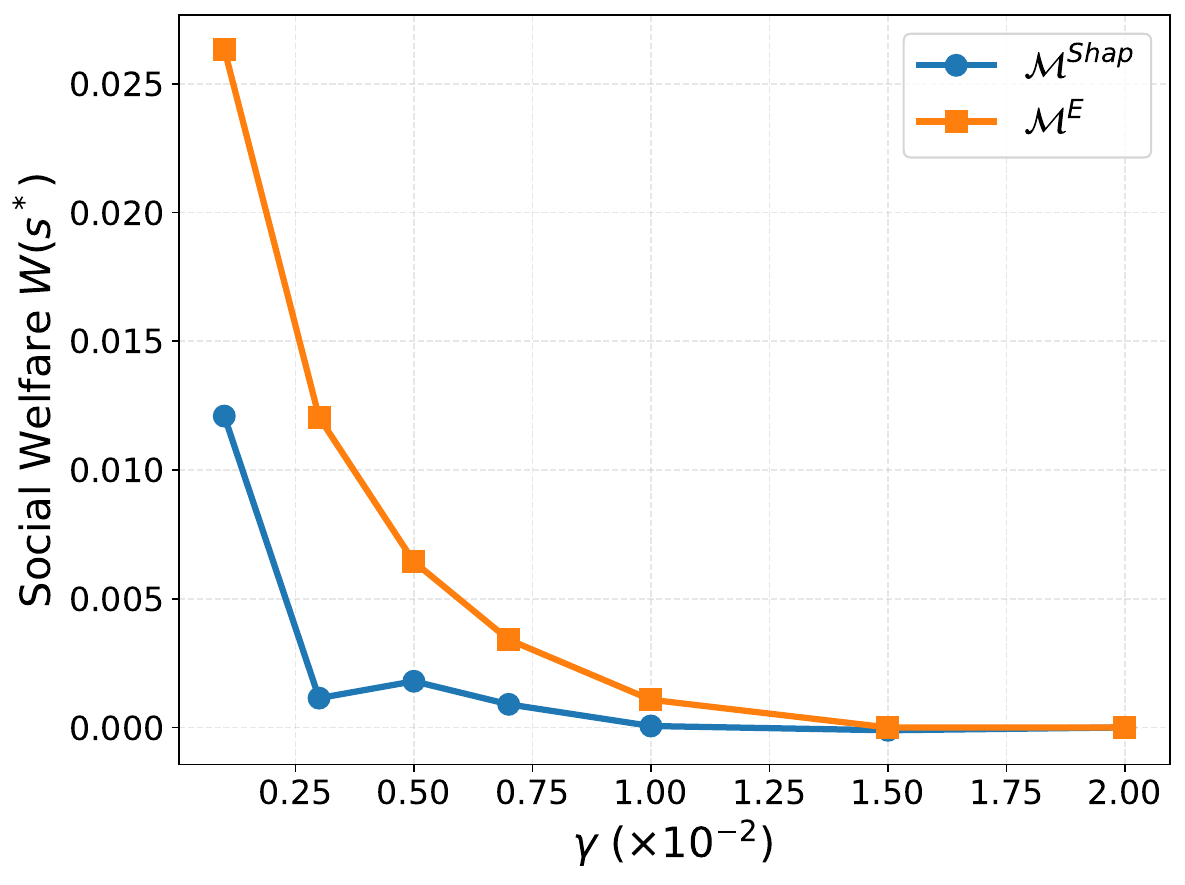}
        \caption{$\beta=0.005$}
        \label{fig:so_beta_0.005}
    \end{subfigure}
    \hfill
    \begin{subfigure}[b]{0.32\linewidth}
        \centering
        \includegraphics[width=\linewidth]{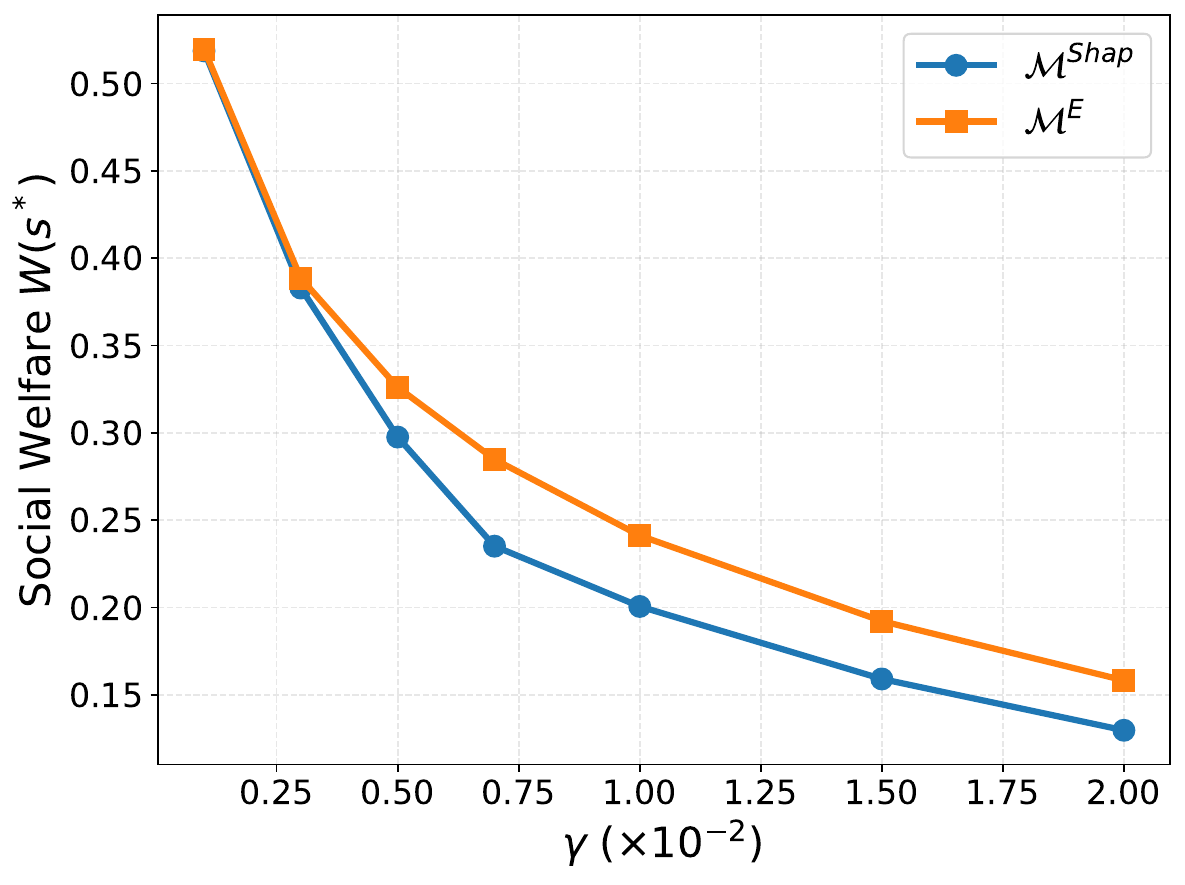}
        \caption{$\beta=0.05$}
        \label{fig:so_beta_0.05}
    \end{subfigure}
    \begin{subfigure}[b]{0.32\linewidth}
        \centering
        \includegraphics[width=\linewidth]{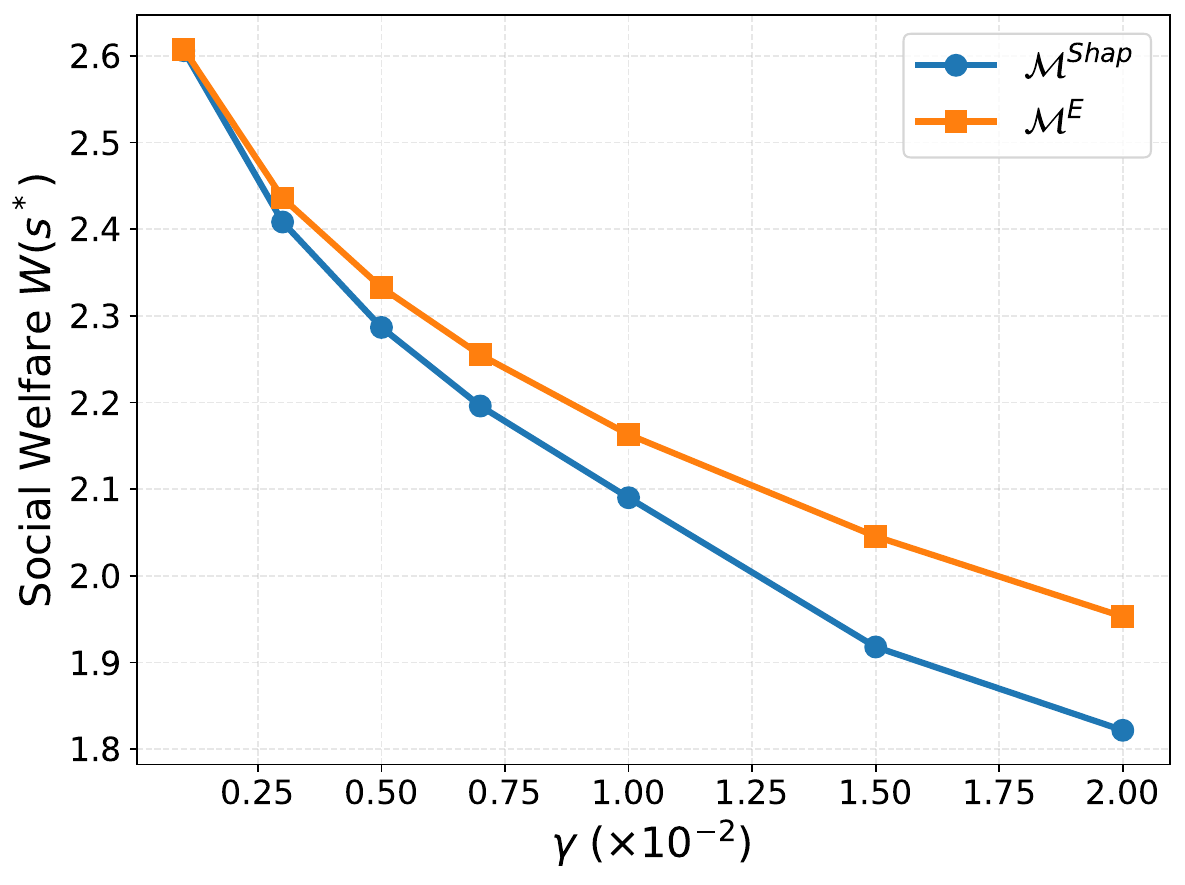}
        \caption{$\beta=0.5$}
        \label{fig:so_beta_0.5}
    \end{subfigure}
    \caption{Social welfare optimality under different $\beta, \gamma$ and $n=3$.}  
    \label{fig:so}
\end{figure}

\subsubsection{Individual Rationality in Heterogeneous Settings}
\label{subsec:heterogeneous_ir}

We evaluate IR under heterogeneity using the weight matrix $W$ for agents indexed $\{0, 1, 2\}$:
\begin{equation}
W =
\begin{bmatrix}
1.0 & w_{01} & 2.0 \\
w_{10} & 1.0 & 0.5 \\
1.0 & 3.0 & 1.0
\end{bmatrix}.
\end{equation}
We fix $\gamma_i=\gamma=0.001$ for all $i\in\mathcal{N}$ and $\beta=0.02$, while varying $w_{01}$ and $w_{10}$ from $-3$ to $3$. We allow $W_{ij}$ to take both positive and negative values: a positive (resp. negative) value indicates that agent $i$ benefits (resp. harms) agent $j$, with larger magnitudes corresponding to stronger effects. The behavior of both mechanisms under negative weights is further analyzed in Section~\ref{sec:discussion}. 

Fig.~\ref{fig:heterogeneous_ir_heatmap_all} shows that $\mathcal{M}^{Shap}$ does not always satisfy IR in heterogeneous settings, as evidenced by  several counterexamples that demonstrate  Theorem \ref{thm:IR_ext} cannot be extended to $\mathcal{M}^{Shap}$ under data heterogeneity. Negative IR gaps arise \emph{even in the positive quadrant}, where all interactions are beneficial. This occurs when agents contribute asymmetrically: an agent may provide substantial value to others but receive insufficient compensation relative to its participation cost. Since $\mathcal{M}^{Shap}$ allocates rewards based on marginal contributions and is independent of individual costs, the resulting payment may not fully offset an agent's cost, leading to a negative IR gap. For example, when $w_{01}$ is low and $w_{10}$ is high\footnote{This can happen if Agent 1 is more privacy sensitive and injects more noise into its data.}, Agent~0 contributes substantial value to Agent~1 but receives insufficient compensation, resulting in $u_0^{Shap}<v_0^*$. As seen in the top row of Fig.~\ref{fig:heterogeneous_ir_heatmap_all}, these positive-quadrant violations affect Clients~0 and~1. In contrast, $\mathcal{M}^{E}$ maintains a strictly positive IR gap throughout the entire parameter space, consistent with Theorem~\ref{thm:IR_ext}.

\begin{figure}[t]
    \centering  
\includegraphics[width=0.95\linewidth]{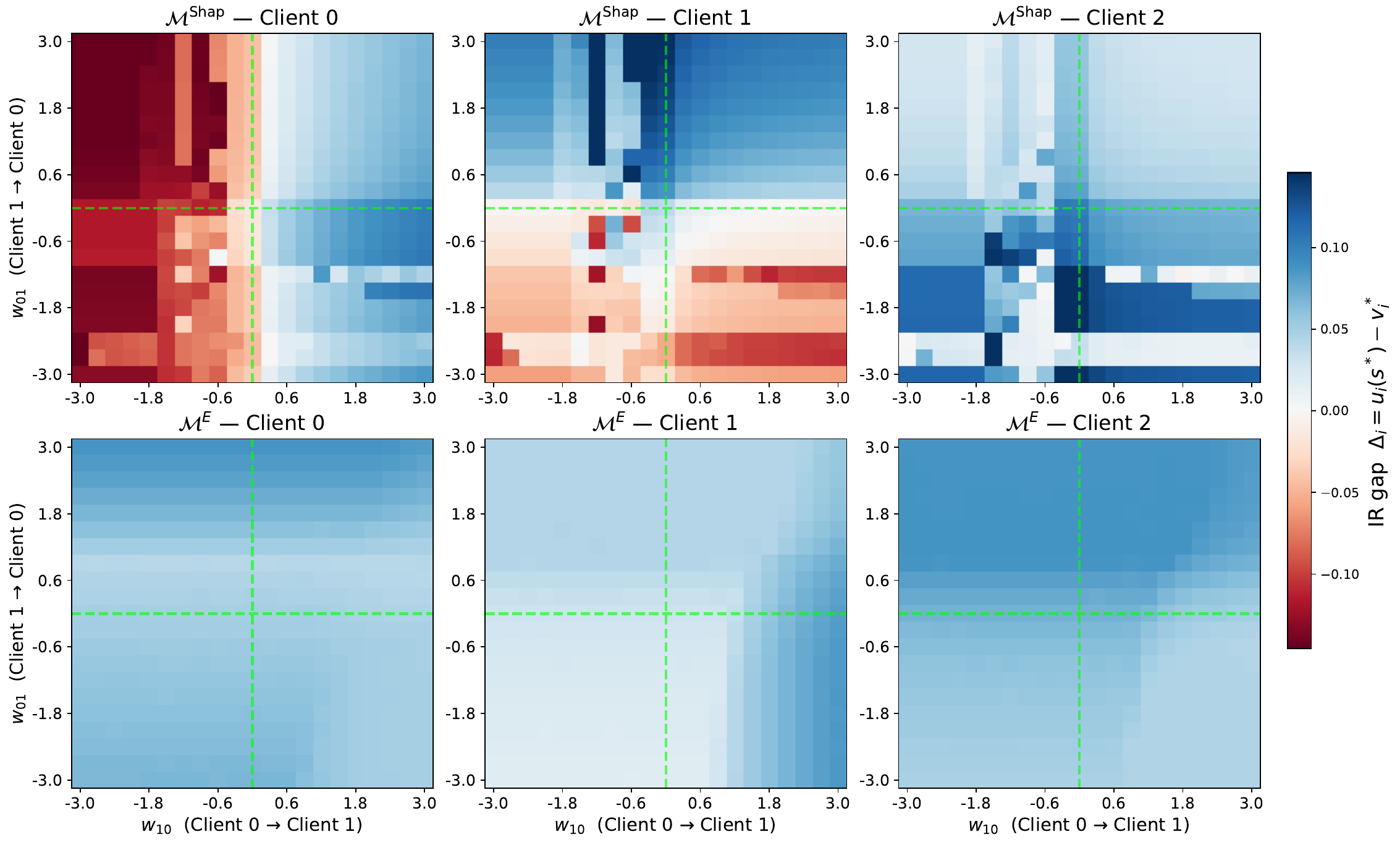}
    \caption{Heatmap of clients' IR gaps ($\Delta_i$) under varying cross-client weights $w_{01}$ and $w_{10}$. The red region denotes areas where IR is violated ($\Delta_i < 0$), while the blue region denotes areas where IR holds ($\Delta_i \ge 0$).}
    
    \label{fig:heterogeneous_ir_heatmap_all}
\end{figure}

\subsubsection{Reciprocity}
\label{subsec:fairness_analysis}

To evaluate the reciprocity of $\mathcal{M}^{E}$, we compute the Shapley values induced by its equilibrium allocation $\mathbf{s}^*=[s_1^*,s_2^*,\ldots,s_n^*]$. Specifically, substituting $\mathbf{s}^*$ into the utility model yields the corresponding client values $a_i$, from which the Shapley value $\varphi_i^{A}$ can be computed. 
Table~\ref{tab:mechanism_comparison_W} compares the reciprocity ratios of both mechanisms under different weight matrices with $\gamma=0.001$ and $\beta=0.02$. 
The columns ``Agent $i$'' report $
\min_{\mathbf{s}\in \mathrm{NE}(\mathcal{M})}
\frac{a_i(\mathbf{s}) + p_i(\mathbf{s})}
{\varphi_i^A(\mathbf{s})}$.
A value of 1 indicates perfect reciprocity. Values greater (resp. less) than 1 indicate that the benefit received by the agent exceeds (resp. falls below) its allocated share.

In the homogeneous setting, $\mathcal{M}^{Shap}$ and $\mathcal{M}^{E}$ achieve perfect reciprocity, consistent with Theorem \ref{thm:reciprocity_ME_homogeneous}. In heterogeneous environments, $\mathcal{M}^{E}$ exhibits noticeable deviations. This behavior stems from its use of personalized marginal pricing to internalize externalities and maximize social welfare. In contrast, $\mathcal{M}^{Shap}$  maintains perfect reciprocity by 
allocating rewards according to marginal contributions. These results highlight the fairness--efficiency trade-off: $\mathcal{M}^{Shap}$ prioritizes fairness, while $\mathcal{M}^{E}$ sacrifices fairness to achieve higher social efficiency.

\begin{table*}[h]
\centering
\caption{Reciprocity under Different Weight Matrices.}

\label{tab:mechanism_comparison_W}
\footnotesize
\begin{tabular}{llcccc}
\toprule
\textbf{Weight Matrix} &
\textbf{Mechanism} &
\textbf{Reciprocity} &
\textbf{Agent 0} &
\textbf{Agent 1} &
\textbf{Agent 2} \\
\midrule

\multirow{2}{*}{
$
W_1=
\begin{bmatrix}
1 & 1 & 1\\
1& 1 &1\\
1 & 1& 1
\end{bmatrix}
$
}
\\
& $\mathcal{M}^{\mathrm{Shap}}$
& 1.000 & 1.000 & 1.000 & 1.000 \\

& $\mathcal{M}^{E}$
& 1.000 & 1.000 & 1.000 & 1.000\\

\midrule

\multirow{2}{*}{
$
W_2=
\begin{bmatrix}
1 & 1.2 & 0.4\\
0.5& 1 & 3\\
1.4 & 0.8& 1
\end{bmatrix}
$
}
\\
& $\mathcal{M}^{\mathrm{Shap}}$
& 1.000 & 1.000 & 1.000 & 1.000 \\

& $\mathcal{M}^{E}$
&0.801 & 0.995 & 1.390 & 0.801  \\

\midrule

\multirow{2}{*}{
$
W_3=
\begin{bmatrix}
1 & 0.1 & 2\\
-0.2 & 1 & 0.5\\
1 & 3 & 1
\end{bmatrix}
$
}
\\
& $\mathcal{M}^{\mathrm{Shap}}$
& 1.000 & 1.000 & 1.000 & 1.000 \\

& $\mathcal{M}^{E}$
& 0.678 & 0.915 & 3.375 & 0.678  \\

\bottomrule
\end{tabular}
\end{table*}

\subsection{Simulations (MNIST)}
We now validate our findings in a realistic FL setting where $a_i$ is replaced by actual test accuracy from training a MLP on MNIST, rather than the analytical payoff model \eqref{eq:payoff}. We consider three agents on the MNIST dataset, each holding 400 images evenly distributed across 10 classes. Simulation settings are summarized in Appendix \ref{app:sim}. We examine two cost scenarios corresponding to low ($\gamma = 0.001$) and moderate ($\gamma = 0.01$) per-sample contribution cost. 

\begin{figure}[h]
    \centering
    \begin{subfigure}[t]{0.45\linewidth}
        \centering
        \includegraphics[width=\linewidth]{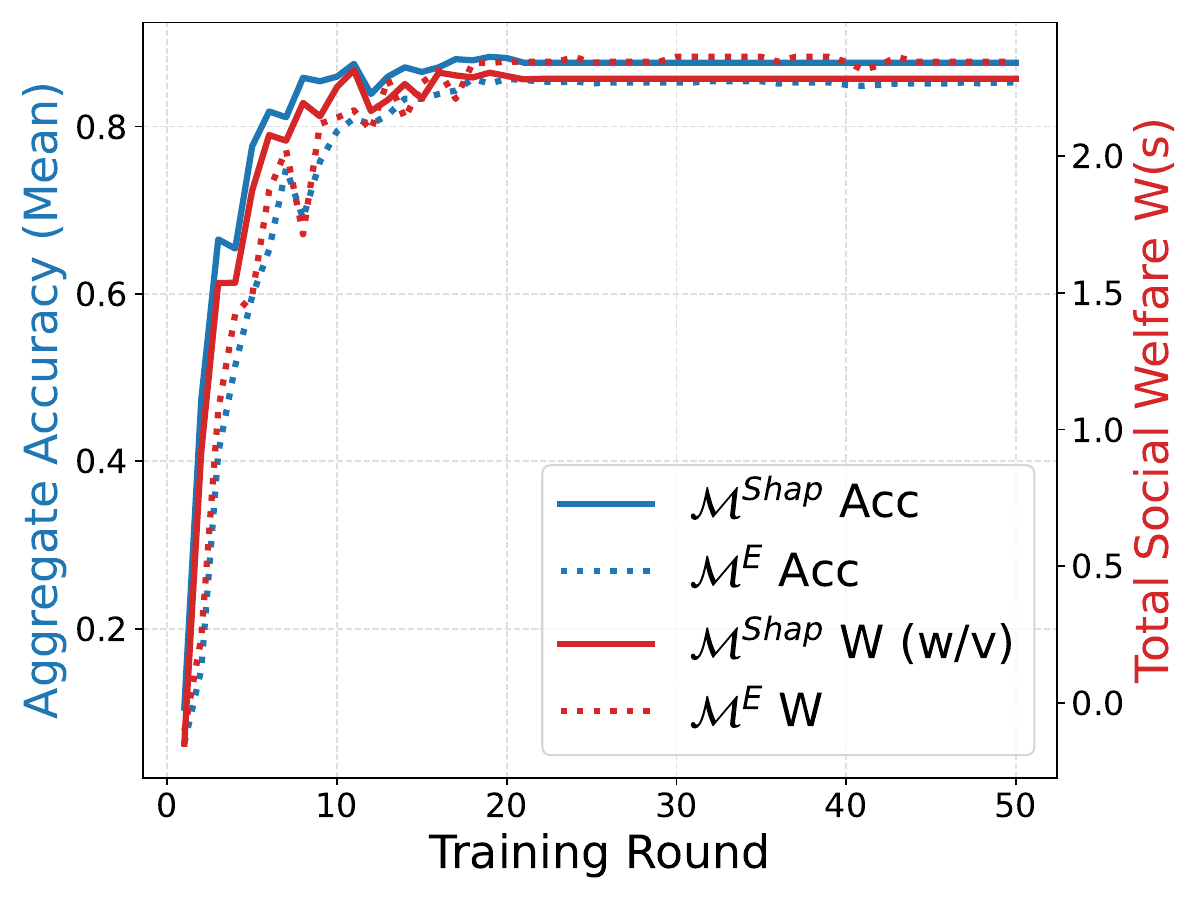}
        \caption{$\gamma=0.001$}
        \label{fig:low_gamma}
    \end{subfigure}
    \hfill
    \begin{subfigure}[t]{0.45\linewidth}
        \centering
        \includegraphics[width=\linewidth]{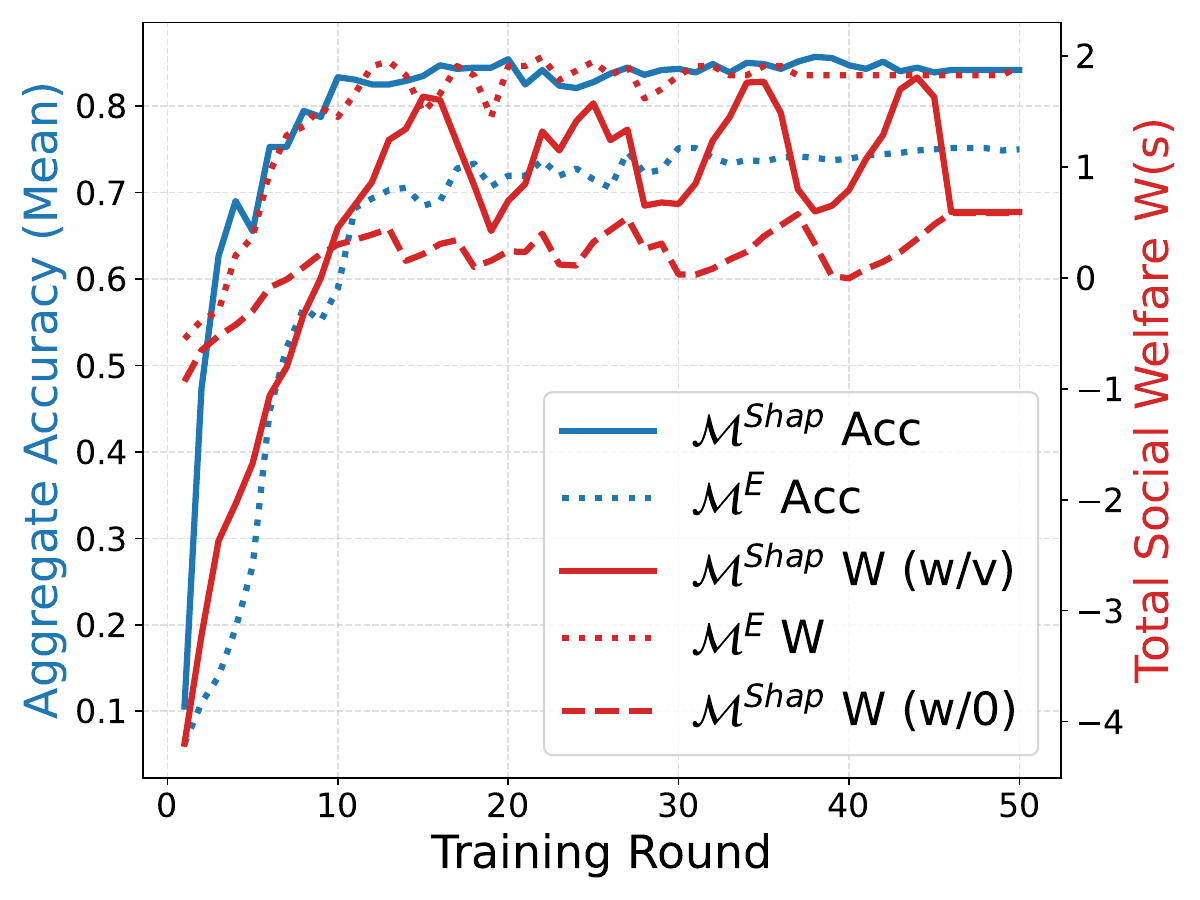}
       \caption{$\gamma=0.01$}
        \label{fig:high_gamma}
    \end{subfigure}
    \caption{Aggregate accuracy and social welfare for $n=3$ under different cost settings, averaged over three random seeds; 
    the legend ``w/v'' denotes welfare calculated by including all agents in the mechanism, whereas ``w/0'' is calculated using standalone utility of those who opt out (under $M^{\text{Shap}}$ in (b)).}
   
    \label{fig:agnews_no_opt}
\end{figure}

In the low-cost scenario (Fig.~\ref{fig:low_gamma}), all agents satisfy IR and participate in the federation under both mechanisms.
Accuracy and welfare
converge to similar values under $\mathcal{M}^{\text{Shap}}$ and $\mathcal{M}^E$, consistent with the theoretical result that both mechanisms guarantee IR under homogeneity (Theorem~\ref{thm:mshap_IR_homo},
Theorem~\ref{thm:IR_ext}). In the moderate-cost scenario (Fig.~\ref{fig:high_gamma}), the marginal accuracy gain no longer justifies the contribution cost for some agents, causing opt-out under $\mathcal{M}^{\text{Shap}}$. While $\mathcal{M}^{\text{Shap}}$ achieves slightly higher aggregate accuracy due to data overcontribution at equilibrium,
$\mathcal{M}^E$ achieves strictly higher social welfare. When opted-out clients are evaluated on their standalone accuracy, $\mathcal{M}^{\text{Shap}}$'s welfare (w/0) drops. The visible fluctuations in $\mathcal{M}^{\text{Shap}}$'s welfare stem from BR dynamics continuously
adjusting contributions. Together, these results confirm that $\mathcal{M}^E$ allocates contributions more efficiently, achieving higher social welfare at lower total data cost, consistent with Theorem~\ref{thm:mshap_SO} and Corollary~\ref{cor:overcontribution}.

\section{Individual Rationality Under Negative Externalities}
\label{sec:discussion}

Throughout our analysis, we have assumed that $w_{ij} \geq 0$ for all $i,j \in N$, meaning no agent's data harms another's model. In practice, however, FL deployments frequently involve agents with conflicting data distributions, incompatible labeling schemes, or misaligned task objectives, where contributing data from one agent can actively degrade model performance for another. For example, in a cross-silo healthcare federation, data collected under different clinical protocols may introduce bias that hurts participating hospitals with different patient populations. Privacy considerations introduce a subtler version of the same phenomenon, where an agent protecting sensitive information may add noise to their data contribution, which, while not adversarial in intent, can reduce the quality of the jointly trained model for other participants. These scenarios motivate relaxing Assumption \ref{assump:weights} to allow $w_{ij}<0$, and raise a natural question: when an agent's data harms others, do the two mechanisms correctly recognize that its participation should not be incentivized?
To answer this question, we extend the IR analysis of Section \ref{sec:IR} to settings where $w_{ij}$ may be negative. Define the \textit{net outgoing externality} of agent $i$ at profile $\vec{s}^*$ as:

\begin{equation}
    T_i(\vec{s}^*) := \sum_{k \neq i} w_{ki} \, \xi_k(\vec{s}^*)^{-\beta - 1},
    \label{eq:Ti}
\end{equation}
where recall $\xi_k(\vec{s}^*)^{-\beta-1}$ is agent $k$'s marginal payoff sensitivity at $\vec{s}^*$. Thus, $T_i(\vec{s}^*)$ measures the total marginal effect of agent $i$'s contribution on all other agents' payoffs at equilibrium, with $T_i(\vec{s}^*) > 0$ (resp. $<0$) indicating a net positive (resp. negative) externality. 

We note that relaxing Assumption~\ref{assump:weights} to allow $w_{ij} < 0$ does not alter the payment rule of $\mathcal{M}^E$, derived solely from the KKT conditions of the welfare maximization problem, which remain valid for any $w_{ij} \in \mathbb{R}$ provided the domain condition $\xi_k(\vec{s}) > 0$ holds for all $\vec{s} \in S$ and $k \in N$. Under non-negative weights this condition is automatic, since $\xi_k(\vec{s})$ is a sum of non-negative terms plus one. Under negative weights it is no longer guaranteed, since sufficiently large contributions from agents who harm $k$ can drive $\xi_k(\vec{s})$ below zero, and we treat it as an explicit restriction on the action space. If this condition fails, the payoff function $a_k$ and consequently the welfare maximization problem are no longer well-defined over all of $S$. When this domain condition is met, what changes under negative weights is the IR guarantee.  

\begin{proposition}
\label{prop:Ti}
Let $w_{ij} \in \mathbb{R}$ for all $i, j \in N$, with $w_{ii} = 1$. Suppose $\xi_k(\vec{s}) > 0$ for all $\vec{s} \in S$ and all $k \in N$, and $\xi_i(\vec{s}^*) \geq 1$ for agent $i$. Under 
Assumptions~\ref{assump:interior} and~\ref{assump:interior_solo}, at any NE $\vec{s}^*$ of $\mathcal{M}^E$ implementing an interior social optimum:
\begin{equation}
    T_i(\vec{s}^*) \geq 0 \iff \Delta_i(\mathcal{M}^E) \geq 0.
\end{equation}
In particular, $T_i(\vec{s}^*) < 0$ implies $\Delta_i(\mathcal{M}^E) < 0$: 
an agent whose contribution harms the federation on net will fail to satisfy individual rationality under $\mathcal{M}^E$.
\end{proposition}

Under $\mathcal{M}^E$, when $T_i(\vec{s}^*) < 0$, the mechanism charges agent $i$ for this harm via the tax term in \eqref{eq:ext_payment_explicit}, reducing their utility below what they could achieve by training alone. We complement Proposition \ref{prop:Ti} with numerical experiments that also examine $\mathcal{M}^\text{Shap}$ under negative weights. We consider a 4-agent instance with $\beta = 0.5$, $\vec{\gamma} = (0.1, 0.3, 0.3, 0.3)$, $\tau = 3$ for all agents $i \in N$, where agents $1,2,3$ fully value each other's data ($w_{jk} = 1$ for $j,k \in \{1,2,3\}$), and agent 0 has outgoing weight $w_{k0}$ and incoming weight $w_{0k}$ to and from each of agents $1,2,3$. We sweep $w_{k0} \in [-1, 0.3]$ for three values of $w_{0k} \in \{-0.1, 0, 0.1\}$, corresponding to agent 0 being passively harmed by, neutral to, and benefiting from the other agents. 

\begin{figure}[htbp]
    \centering
    \begin{subfigure}{0.32\textwidth}
        \includegraphics[width=\textwidth]{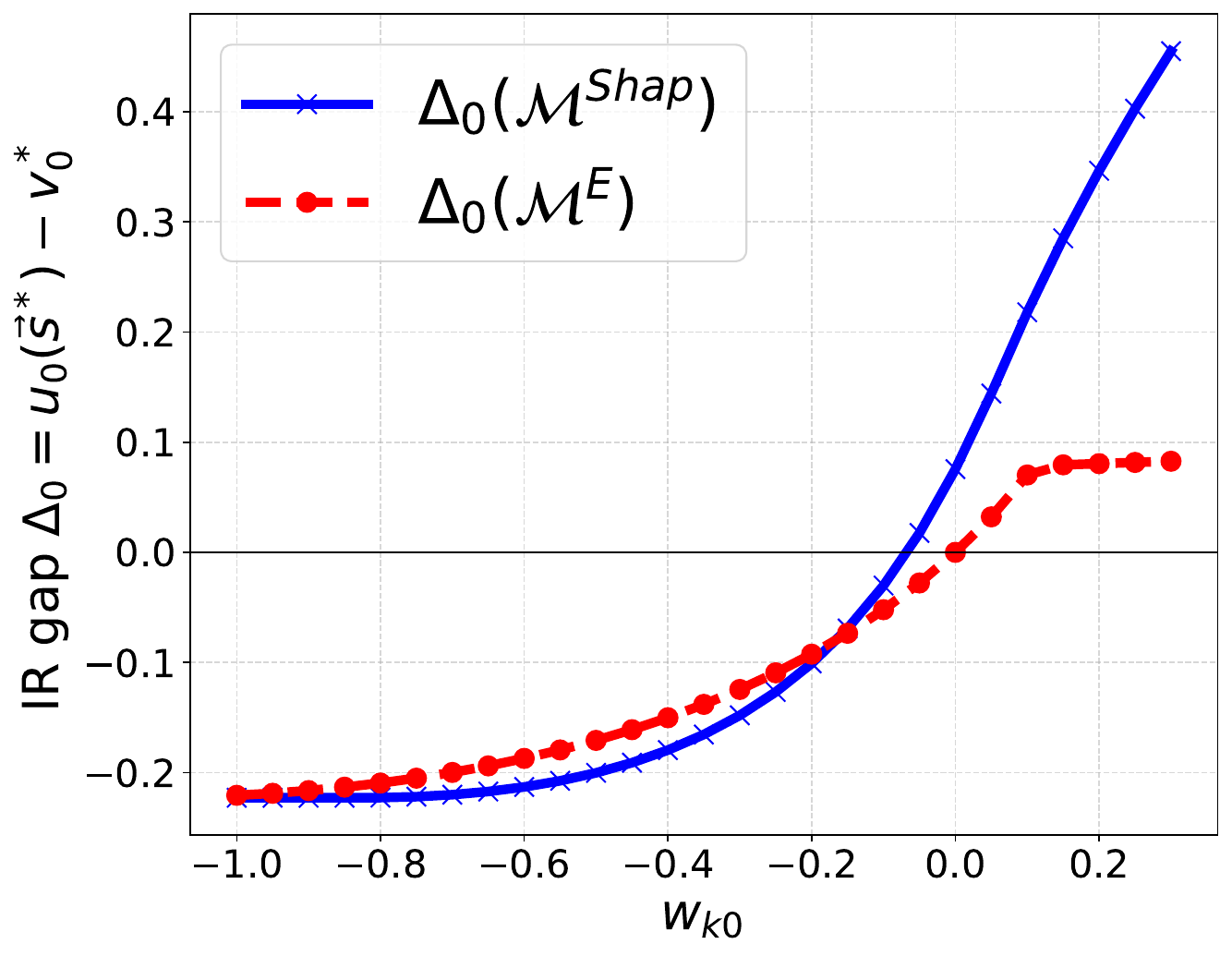}
        \caption{$w_{0k} = -0.1$}
        \label{fig:ir_neg}
    \end{subfigure}
    \hfill
    \begin{subfigure}{0.32\textwidth}
        \includegraphics[width=\textwidth]{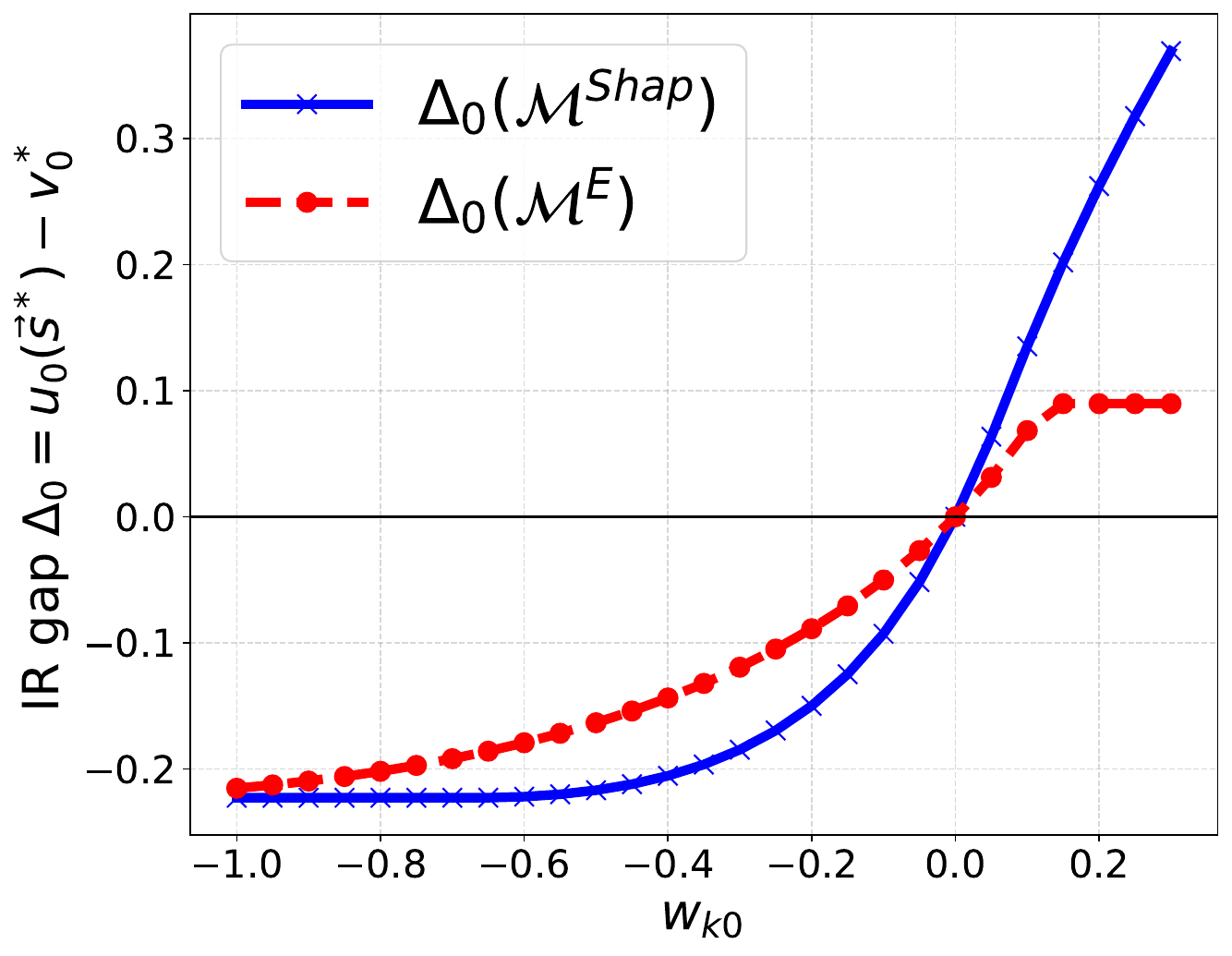}
        \caption{$w_{0k} = 0$}
        \label{fig:ir_zero}
    \end{subfigure}
    \hfill
    \begin{subfigure}{0.32\textwidth}
        \includegraphics[width=\textwidth]{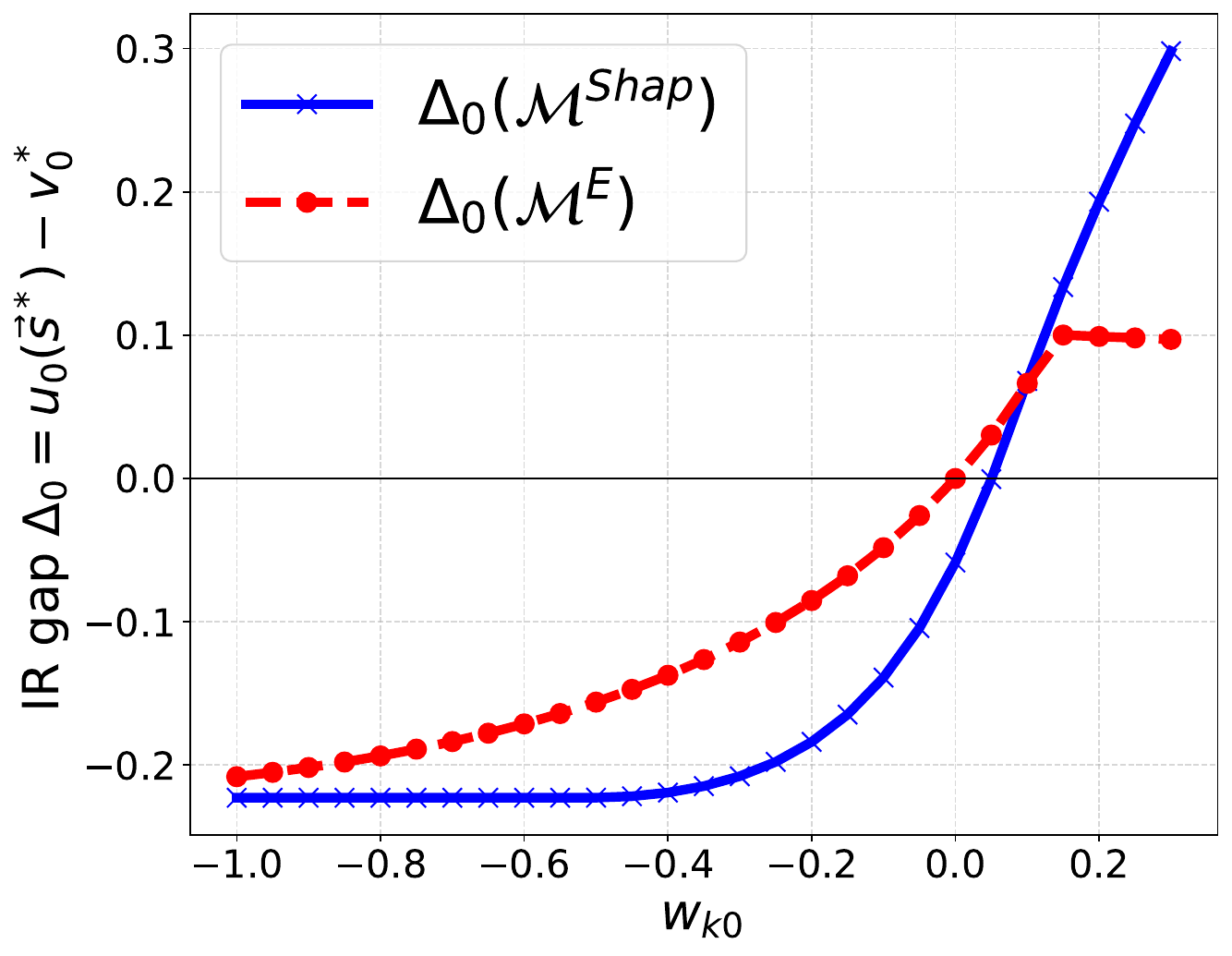}
        \caption{$w_{0k} = 0.1$}
        \label{fig:ir_pos}
    \end{subfigure}
        \caption{IR gaps for agent 0 ($\Delta_0$) vs. outgoing weight $w_{k0}$, for 
    $w_{0k} \in \{-0.1,\, 0,\, 0.1\}$.}
    \label{fig:ir_negative_weights}
\end{figure}

Figure~\ref{fig:ir_negative_weights} reports the IR gaps $\Delta_0(\mathcal{M}^E)$ and $\Delta_0(\mathcal{M}^{Shap})$ as a function of $w_{k0}$, with more negative values indicating greater harm. The zero crossing of each curve marks the threshold below which the mechanism fails to retain agent 0. In Figure \ref{fig:ir_neg}, agent 0 is passively harmed by others ($w_{0k} = -0.1$). Consistent with Proposition~\ref{prop:Ti}, $\mathcal{M}^E$'s zero crossing remains at $w_{k0} = 0$.  Under $\mathcal{M}^{Shap}$, the zero crossing shifts to a more negative $w_{k0}$, meaning $\mathcal{M}^{Shap}$ retains agent 0 even when it causes mild harm to others. In this case, empirically, the payment $p_0^{Shap}$ is positive throughout the entire range of $w_{k0}$, meaning $\mathcal{M}^{Shap}$ pays agent 0 to compensate for the incoming harm that decreases $a_0(\vec{s}^*)$ below $\varphi_0^A(\vec{s}^*)$. This positive payment shifts the $\Delta_0$ curve, moving $\mathcal{M}^{Shap}$'s zero crossing to a more negative $w_{k0}$. In Figure~\ref{fig:ir_zero}, $w_{0k} = 0$, so agent 0 neither benefits from nor is harmed by others' data. Both mechanisms' IR gaps cross zero at $w_{k0} = 0$. In Figure~\ref{fig:ir_pos}, agent 0 passively benefits from others ($w_{0k} = 0.1$). Under $\mathcal{M}^{Shap}$, the zero crossing shifts to a more positive $w_{k0}$, meaning $\mathcal{M}^{Shap}$ requires agent 0 to provide a strictly positive benefit to others before retaining them. This is because the incoming benefit inflates $a_0(\vec{s}^*)$ above $\varphi_0^A(\vec{s}^*)$, generating a negative payment $p_0^{Shap}$ that shifts $\mathcal{M}^{Shap}$'s zero crossing to a more positive $w_{k0}$. 

Together, these results illustrate that $\mathcal{M}^E$ and $\mathcal{M}^{Shap}$ respond to data heterogeneity differently: $\mathcal{M}^E$ conditions participation on the externality an agent imposes on others, while $\mathcal{M}^{Shap}$ additionally accounts for the passive benefit an agent receives from the federation, making its participation standard sensitive to both the incoming and outgoing structure of the data relationship.

\section{Conclusion and Future Work}\label{sec:conclusion}
This paper establishes three main results: $\mathcal{M}^{Shap}$ achieves perfect reciprocity but fails social optimality at every interior equilibrium and can fail IR under data heterogeneity, while $\mathcal{M}^E$ achieves social optimality and guarantees IR under non-negative data heterogeneity, but sacrifices reciprocal fairness.

Thus, $\mathcal{M}^{Shap}$ is preferable when reciprocal fairness is the primary concern. It guarantees each agent's net benefit reflects their marginal contribution, and is well suited to homogeneous settings where IR is guaranteed. $\mathcal{M}^E$ is preferable when social welfare maximization is the priority. It implements the 
socially optimal contribution profile by design and maintains IR under arbitrary non-negative data heterogeneity, at the cost of sacrificing reciprocal fairness. We note that in these settings where IR is guaranteed, the absence of reciprocity may be a reasonable tradeoff: agents are already better off participating than training alone, so while some may receive disproportionately higher rewards than their contributions warrant, participation remains individually rational for all.

Under negative externalities, $\mathcal{M}^E$ is slightly advantaged: when an agent's data harms the federation on net, IR fails for $\mathcal{M}^E$ by Proposition~\ref{prop:Ti}, correctly identifying this agent as one whose participation should not be incentivized. $\mathcal{M}^{Shap}$'s behavior in this case is more nuanced, as its participation criterion depends on both the outgoing and incoming structure of the data relationship, and a general analytical characterization remains open.

Our analysis and simulations assume that all data samples contribute equally to model performance. In practice, however, individual data may have heterogeneous effects on accuracy. This is also a future direction worth pursuing.

\medskip
\noindent\textbf{Disclosure of Interests.} The authors have no competing interests to declare that are relevant to the content of this article.

\bibliographystyle{splncs04}
\bibliography{references}
\appendix
\renewcommand{\theHsection}{appendix.\thesection}

\section{Derivation of the Externality Mechanism for FL}
\label{app:ext_deriv}

We derive the payment rule \eqref{eq:ext_payment_explicit} 
following the approach of 
\cite{sharma2012localpublicgoodprovisioning}. Every Nash 
equilibrium $\vec{s}^*$ of $\mathcal{M}^E$ is the 
socially optimal contribution profile, i.e. the 
maximizer of total welfare $W(\vec{s})$ over 
$\mathcal{S}$:
\begin{equation}
     \quad \max_{\vec{s}} 
    \sum_{i \in N} \left[ a_i(\vec{s}) - c_i(s_i) 
    \right] \quad \text{s.t.} \quad s_i \in [0, \tau_i] 
    \quad \forall i \in N.
    \label{eq:PC}
\end{equation}
The KKT conditions of this optimization 
problem yield the personalized prices $l^*_{ij}$, from 
which the payment rule follows directly.

\subsection{KKT Conditions of the Social Optimum}

Since the objective in \eqref{eq:PC} is concave in $\vec{s}$ and 
$\mathcal{S}$ is convex and compact, the KKT conditions are necessary and 
sufficient for optimality. Taking the gradient of the Lagrangian with respect to  $s_i$ and setting it to zero gives the stationarity condition for each $i \in N$ :

\begin{equation}
    \beta \sum_{k \in N} w_{ki} \xi_k(\vec{s}^*)^{-\beta-1} - \gamma_i 
    + \lambda_i - \mu_i = 0,
    \label{eq:KKT_stat}
\end{equation}

where $\lambda_i \geq 0$ and $\mu_i \geq 0$ are the KKT multipliers for the 
constraints $s_i \geq 0$ and $s_i \leq \tau_i$ respectively, with 
complementary slackness conditions $\lambda_i s^*_i = 0$ and 
$\mu_i(s^*_i - \tau_i) = 0$. Under Assumption \ref{assump:interior}, 
$\lambda_i = \mu_i = 0$ for all $i \in N$, so \eqref{eq:KKT_stat} simplifies 
to:
\begin{equation}
    \beta \sum_{k \in N} w_{ki} \xi_k(\vec{s}^*)^{-\beta-1} = \gamma_i, 
    \quad \forall i \in N.
    \label{eq:KKT_interior}
\end{equation}

\subsection{Derivation of Personalized Prices}

Following \cite{sharma2012localpublicgoodprovisioning}, we define the personalized price $l_{ij}^*$ as the marginal utility that agent $i$ derives from agent $j$'s contribution at $\vec{s}^*$, i.e. $l^*_{ij} = 
\frac{\partial u_i(\vec{s})}{\partial s_j}\big|_{\vec{s}^*}$. Since 
$u_i(\vec{s}) = a_i(\vec{s}) - c_i(s_i)$ and $c_i(s_i) = \gamma_i s_i$ 
depends only on $s_i$, for $j \neq i$:

\begin{equation}
    l^*_{ij} = \frac{\partial a_i(\vec{s})}{\partial s_j}\bigg|_{\vec{s}^*}
    = \beta w_{ij} \xi_i(\vec{s}^*)^{-\beta-1}, \quad j \neq i,
\end{equation}
where $\frac{\partial \xi_i}{\partial s_j} = w_{ij}$ given the definition of $\xi$.
For $j = i$, the price additionally accounts for agent 
$i$'s marginal cost:
\begin{equation}
    l^*_{ii} = \frac{\partial a_i(\vec{s})}{\partial s_i}\bigg|_{\vec{s}^*} 
    - \gamma_i = \beta \xi_i(\vec{s}^*)^{-\beta-1} - \gamma_i,
\end{equation}
where we used $w_{ii} = 1$ and the fact that $\lambda_i = \mu_i = 0$ under Assumption
\ref{assump:interior}. Together these match \eqref{eq:prices}.

\subsection{The $\mathcal{M}^E$ Payment Rule}

Substituting the personalized prices into $p_i^E(\vec{s}^*) = 
-\sum_{j \in N} l^*_{ij} s^*_j$ and splitting into $j = i$ and $j \neq i$:
\begin{align}
    p_i^E(\vec{s}^*) 
    &= -l^*_{ii} s^*_i - \sum_{j \neq i} l^*_{ij} s^*_j \notag \\
    &= -\left(\beta \xi_i(\vec{s}^*)^{-\beta-1} - \gamma_i\right) s^*_i 
    - \sum_{j \neq i} \beta w_{ij} \xi_i(\vec{s}^*)^{-\beta-1} s^*_j 
    \notag \\
    &= \gamma_i s^*_i - \beta \xi_i(\vec{s}^*)^{-\beta-1} 
    \left(s^*_i + \sum_{j \neq i} w_{ij} s^*_j\right) \notag \\
    &= \gamma_i s^*_i - \beta \xi_i(\vec{s}^*)^{-\beta-1} 
    \sum_{j \in N} w_{ij} s^*_j \notag \\
    &= \gamma_i s^*_i - \beta \xi_i(\vec{s}^*)^{-\beta-1} 
    \left(\xi_i(\vec{s}^*) - 1\right),
    \label{eq:payment_derived}
\end{align}
where the last step uses $\sum_{j \in N} w_{ij} s^*_j = \xi_i(\vec{s}^*) 
- 1$, establishing \eqref{eq:ext_payment_explicit}.

\section{Supplementaries of Section 4}
\label{app:mshap_SO}
\subsection{Proof of Theorem \ref{thm:mshap_SO}}

\begin{proof}
Suppose for contradiction that $\vec{s}^*$ is an interior Nash equilibrium of $\mathcal{M}^{\text{Shap}}$ 
that maximizes total welfare $W(\vec{s})$.

\medskip
\noindent\textbf{Step 1: FOC at the NE of $\mathcal{M}^{\text{Shap}}$.}
Under $\mathcal{M}^{\text{Shap}}$, $p^{\text{Shap}}_i(\vec{s}) = \varphi^A_i(\vec{s}) - a_i(\vec{s})$ for all $\vec{s}$, so agent $i$'s utility is $u^{\text{Shap}}_i(\vec{s}) =\varphi^A_i(\vec{s}) - \gamma_i s_i$ for all $\vec{s}$. Since $\varphi^A_i$ is concave in $s_i$ (each term $A(\vec{s}[X \cup \{i\}])$ is concave in $s_i$ since $\xi_j$ is linear in $s_i$ and $-\xi_j^{-\beta}$ is concave; $\varphi^A_i$ is then concave as a positive weighted sum) and the NE is interior, the FOC is necessary and sufficient:
\begin{equation}
\label{eq:SO_FOC1}
    \frac{\partial \varphi^A_i}{\partial s_i}(\vec{s}^*) = \gamma_i \quad \forall i \in N. 
\end{equation}

\noindent\textbf{Step 2: FOC at the social optimum.} Since $W(\vec{s})$ is concave (each $a_i(\vec{s}) = 1 - \xi_i(\vec{s})^{-\beta}$ 
is concave in $\vec{s}$ and $-\gamma_i s_i$ is linear, so $W$ is concave as a sum of concave functions) and the maximizer is interior by hypothesis:
\begin{equation}
\label{eq:SO_FOC2}
    \frac{\partial A}{\partial s_i}(\vec{s}^*) = \sum_{j \in N} \beta w_{ji} \xi_j(\vec{s}^*)^{-\beta-1} = \gamma_i \quad \forall i \in N. 
\end{equation}

\noindent\textbf{Step 3: Necessary condition for simultaneous achievability.} Since both \eqref{eq:SO_FOC1} and \eqref{eq:SO_FOC2} equal $\gamma_i$:
\begin{equation}
\label{eq:SO_3}
    \frac{\partial \varphi^A_i}{\partial s_i}(\vec{s}^*) = \frac{\partial A}{\partial s_i}(\vec{s}^*) \quad \forall i \in N. 
\end{equation}

\noindent\textbf{Step 4: Expanding $\frac{\partial \varphi^A_i}{\partial s_i}$.} Differentiating \eqref{eq:shapley} with respect to $s_i$. Since $i \notin X$ for all $X \subseteq N \setminus \{i\}$, we have $\frac{\partial}{\partial s_i} A(\vec{s}[X]) = 0$. Therefore:
\begin{equation}
\label{eq:SO_4}
    \frac{\partial \varphi^A_i}{\partial s_i}(\vec{s}) = \sum_{X \subseteq N \setminus \{i\}} w_X \cdot \frac{\partial A}{\partial s_i}(\vec{s}[X \cup \{i\}])
\end{equation}
where $w_X = \frac{|X|!(n-|X|-1)!}{n!} > 0$ and $\sum_{X \subseteq N \setminus \{i\}} w_X = 1$, so \eqref{eq:SO_4} is a weighted average of $\frac{\partial A}{\partial s_i}$ across sub-coalition profiles.

\medskip
\noindent\textbf{Step 5: Strict inequality.}
For any $X \subseteq N \setminus \{i\}$, agents outside $X \cup \{i\}$ contribute zero in $\vec{s}[X \cup \{i\}]$, so by Assumption \ref{assump:weights} ($w_{jk} \geq 0$), for all $j \in N$:
\begin{equation}
\label{eq:SO_5a}
    \xi_j(\vec{s}[X \cup \{i\}]) = \sum_{k \in X \cup \{i\}} w_{jk} s_k + 1 \leq \sum_{k \in N} w_{jk} s_k + 1 = \xi_j(\vec{s}).
\end{equation}
Since $x \mapsto x^{-\beta-1}$ is strictly decreasing on $(0, \infty)$ with derivative $-(\beta+1)x^{-\beta-2} < 0$, and $\xi_j(\vec{s}) \geq 1 > 0$ for all $\vec{s} \in \mathcal{S}$:
\begin{equation}
\label{eq:SO_5b}
    \frac{\partial A}{\partial s_i}(\vec{s}[X \cup \{i\}]) \geq \frac{\partial A}{\partial s_i}(\vec{s}) \quad \forall X \subseteq N \setminus \{i\}.
\end{equation}
For strict inequality, consider $X = \emptyset$. By Assumption \ref{assump:weights}, at least three agents derive positive value from agent $i$'s data, so there exists $j \neq i$ with $w_{ji} > 0$; fix such a $j$. Since $w_{jj} = 1$ and $s_j > 0$ because $\vec{s}$ is interior:
\begin{equation}
\label{eq:SO_5c}
    \xi_j(\vec{s}[\{i\}]) = w_{ji} s_i + 1 < \sum_{k \in N} w_{jk} s_k + 1 = \xi_j(\vec{s}),
\end{equation}
since $\sum_{k \neq i} w_{jk} s_k \geq w_{jj} s_j = s_j > 0$. Since $w_{ji} > 0$, the $j$-th term in $\frac{\partial A}{\partial s_i}$ is strictly larger at $\vec{s}[\{i\}]$ than at $\vec{s}$, giving $\frac{\partial A}{\partial s_i}(\vec{s}[\{i\}]) > \frac{\partial A}{\partial s_i}(\vec{s})$. Since $w_\emptyset = \frac{0!(n-1)!}{n!} = \frac{1}{n} > 0$, the weighted average in \eqref{eq:SO_4} strictly exceeds $\frac{\partial A}{\partial s_i}(\vec{s})$ at any interior $\vec{s}$:
\begin{equation}
\label{eq:SO_5d}
    \frac{\partial \varphi^A_i}{\partial s_i}(\vec{s}) > \frac{\partial A}{\partial s_i}(\vec{s}) \quad \forall \vec{s} \text{ interior}, \quad \forall i \in N.
\end{equation}

\noindent\textbf{Step 6: Contradiction.}
\eqref{eq:SO_5d} states that the strict inequality holds at every interior point, so \eqref{eq:SO_3} cannot hold at $\vec{s}^*$. This contradicts our assumption that $\vec{s}^*$ simultaneously satisfies \eqref{eq:SO_FOC1} and \eqref{eq:SO_FOC2}. Therefore no interior Nash equilibrium of $\mathcal{M}^{\text{Shap}}$ maximizes $W(\vec{s})$. \qed
\end{proof}

\subsection{Proof of Corollary \ref{cor:overcontribution}}
\begin{proof}
Let $\vec{s}^{\circ}$ be any welfare maximizer. By Step~1 of the proof of \Cref{thm:mshap_SO}, which uses only that the equilibrium is interior, \eqref{eq:SO_FOC1} holds at $\vec{s}^{NE}$; combined with \eqref{eq:SO_5d}, for every $i \in N$:
\begin{equation}
\label{eq:gradW_neg}
    \frac{\partial W}{\partial s_i}(\vec{s}^{NE})
    = \frac{\partial A}{\partial s_i}(\vec{s}^{NE}) - \gamma_i
    = \frac{\partial A}{\partial s_i}(\vec{s}^{NE})
    - \frac{\partial \varphi_i^A}{\partial s_i}(\vec{s}^{NE}) < 0.
\end{equation}

Suppose for contradiction that $\vec{s}^{\circ} \geq \vec{s}^{NE}$ componentwise. Since $W$ is concave,
\begin{equation*}
    W(\vec{s}^{\circ}) \leq W(\vec{s}^{NE})
    + \nabla W(\vec{s}^{NE})^\top (\vec{s}^{\circ} - \vec{s}^{NE})
    \leq W(\vec{s}^{NE}),
\end{equation*}
where the second inequality uses \eqref{eq:gradW_neg} and
$\vec{s}^{\circ} - \vec{s}^{NE} \geq 0$. By \Cref{thm:mshap_SO}, $\vec{s}^{NE}$ is not a welfare maximizer, so $W(\vec{s}^{\circ}) > W(\vec{s}^{NE})$, a contradiction. Hence $s_i^{\circ} < s_i^{NE}$ for some $i \in N$.

In the homogeneous case, $w_{ij} = 1$ gives $\xi_i(\vec{s}) = \|\vec{s}\|_1 + 1$ for every $i$, and $\gamma_i = \gamma$ gives $\sum_i \gamma_i s_i = \gamma \|\vec{s}\|_1$, so $W(\vec{s}) = h(\|\vec{s}\|_1)$ with $h(z) = n[1 - (z+1)^{-\beta}] - \gamma z$ on $[0, Z]$, $Z := \sum_i \tau_i$. Write $a := \|\vec{s}^{NE}\|_1$ and note $a > 0$ since $\vec{s}^{NE}$ is interior. By \eqref{eq:gradW_neg}, $h'(a) = n\beta(a+1)^{-\beta-1} - \gamma < 0$, and $h'$ is strictly decreasing. Let $z^{\circ} := \|\vec{s}^{\circ}\|_1$, a maximizer of $h$ over $[0, Z]$. If $z^{\circ} = Z$, optimality requires $h'(Z) \geq 0$, impossible since $Z \geq a$ and $h'(Z) \leq h'(a) < 0$. If $z^{\circ} \in (0, Z)$, then $h'(z^{\circ}) = 0 > h'(a)$, and strict monotonicity of $h'$ forces $z^{\circ} < a$. If $z^{\circ} = 0$, then $z^{\circ} < a$ directly. In all cases $\|\vec{s}^{\circ}\|_1 <\|\vec{s}^{NE}\|_1$. \qed
\end{proof}

\section{Supplementaries of Section 5}
\subsection{Proof of Theorem \ref{thm:mshap_IR_homo}}
\label{app:mshap_IR_homo}

\begin{proof}

Fix any agent $i \in N$ and any NE $\vec{s}^*$ of $\mathcal{M}^{\text{Shap}}$. We show $u^{\text{Shap}}_i(\vec{s}^*) \geq v^*_i$.

\medskip
\noindent\textbf{Step 1: NE definition. } 
Since $\vec{s}^*$ is a Nash equilibrium and $s^{\text{solo}}_i \in S_i$ by definition, agent $i$ cannot improve by deviating to $s^{\text{solo}}_i$:
\begin{equation}
\label{eq:IR_NE}
    u^{\text{Shap}}_i(\vec{s}^*) \geq 
    u^{\text{Shap}}_i(s^{\text{solo}}_i, \vec{s}^*_{-i}).
\end{equation}

\noindent\textbf{Step 2: Expand right side. }
Since $p^{\text{Shap}}_i(\vec{s}) = \varphi^A_i(\vec{s})- a_i(\vec{s})$ for all $\vec{s}$, agent $i$'s utility simplifies to $u^{\text{Shap}}_i(\vec{s}) =\varphi^A_i(\vec{s}) - \gamma s_i$ for all $\vec{s}$. Therefore:

\begin{equation}
\label{eq:IR_expand}
    u^{\text{Shap}}_i(s^{\text{solo}}_i, \vec{s}^*_{-i}) 
    = \varphi^A_i(s^{\text{solo}}_i, \vec{s}^*_{-i}) 
    - \gamma s^{\text{solo}}_i.
\end{equation}

Substituting $v^*_i =a_i(s^{\text{solo}}_i, \mathbf{0}_{-i}) - \gamma s^{\text{solo}}_i  = 1-(s^{\text{solo}}_i+1)^{-\beta}-\gamma s^{\text{solo}}_i$ and cancelling $\gamma s^{\text{solo}}_i$ from both sides, it suffices to show:

\begin{equation}
\label{eq:IR_suffices}
    \varphi^A_i(s^{\text{solo}}_i, \vec{s}^*_{-i}) 
    \geq 
    a_i(s^{\text{solo}}_i, \mathbf{0}_{-i}).
\end{equation}

\noindent\textbf{Step 3: Shapley value lower bound. } From the Shapley formula \eqref{eq:shapley} with $\vec{s} = (s^{\text{solo}}_i, \vec{s}^*_{-i})$:
\begin{equation}
    \varphi^A_i(s^{\text{solo}}_i, \vec{s}^*_{-i}) = 
    \sum_{X \subseteq N \setminus \{i\}} 
    \frac{|X|!(n-|X|-1)!}{n!} 
    [A(\vec{s}[X \cup \{i\}]) - A(\vec{s}[X])].
\end{equation}
We show the empty coalition term equals $a_i(s^{\text{solo}}_i, \mathbf{0}_{-i})$ and all other terms are non-negative.

\medskip
\noindent\textit{Empty coalition term ($X = \emptyset$). } 
We first establish this for general payoff functions satisfying $a_i(\mathbf{0}) = 0$ for all $i \in N$ and $a_j(s, \mathbf{0}_{-i}) = a_i(s, \mathbf{0}_{-i})$ for all $i, j \in N$. Since $A(\mathbf{0}) = 0$, the empty coalition term is $\frac{1}{n} A(s^{\text{solo}}_i, \mathbf{0}_{-i})$. In the general case, since $a_j(s^{\text{solo}}_i, \mathbf{0}_{-i}) = a_i(s^{\text{solo}}_i, \mathbf{0}_{-i})$ for all $j \in N$:

\begin{equation}
    A(s^{\text{solo}}_i, \mathbf{0}_{-i}) = 
    \sum_{j \in N} a_j(s^{\text{solo}}_i, \mathbf{0}_{-i}) 
    = n \cdot a_i(s^{\text{solo}}_i, \mathbf{0}_{-i}).
\end{equation}

\noindent so the empty coalition term equals 
$\frac{1}{n} \cdot n \cdot a_i(s^{\text{solo}}_i, 
\mathbf{0}_{-i}) = a_i(s^{\text{solo}}_i, 
\mathbf{0}_{-i})$. In our specific problem formulation, since $A(\mathbf{0}) = 0$ and in the homogeneous case $w_{ji} = 1$ for all $j$:
\begin{equation}
    \xi_j(s^{\text{solo}}_i, \mathbf{0}_{-i}) = 
    s^{\text{solo}}_i + 1 \quad \forall j \in N,
\end{equation}

\noindent so $A(s^{\text{solo}}_i, \mathbf{0}_{-i}) =n[1-(s^{\text{solo}}_i+1)^{-\beta}]$. The empty coalition term is therefore:

\begin{equation}
    \frac{1}{n} \cdot n[1-(s^{\text{solo}}_i+1)^{-\beta}] 
    = 1-(s^{\text{solo}}_i+1)^{-\beta} = 
    a_i(s^{\text{solo}}_i, \mathbf{0}_{-i}).
\end{equation}

\medskip
\noindent\textit{All other terms non-negative ($X \neq \emptyset$). } We first establish this for general payoff functions satisfying $a_j$ non-decreasing in $s_i$. In the general case, since adding agent $i$'s contribution $s^{\text{solo}}_i \geq 0$ to any coalition $X$ weakly increases $a_j$ for all $j \in N$ by monotonicity:
\begin{equation}
\label{eq:A_diff_geq_0}
    A(\vec{s}[X \cup \{i\}]) - A(\vec{s}[X]) 
    \geq 0 \quad \forall X \subseteq N \setminus 
    \{i\}.
\end{equation}

In our specific formulation, for any $X \subseteq N \setminus \{i\}$, adding agent $i$'s contribution $s^{\text{solo}}_i$ to coalition $X$ increases $\xi_j$ by $w_{ji} s^{\text{solo}}_i = s^{\text{solo}}_i \geq 0$ for all $j \in N$, where $s^{\text{solo}}_i \geq 0$ since $s^{\text{solo}}_i \in S_i$. Since $\frac{\partial a_j}{\partial \xi_j} = \beta 
\xi_j^{-\beta-1} > 0$, each term $a_j(\vec{s}[X \cup \{i\}]) - a_j(\vec{s}[X]) 
\geq 0$, so \eqref{eq:A_diff_geq_0} holds. Combining, $\varphi^A_i(s^{\text{solo}}_i, \vec{s}^*_{-i}) \geq a_i(s^{\text{solo}}_i, 
\mathbf{0}_{-i})$, establishing \eqref{eq:IR_suffices}.

\medskip
\noindent\textbf{Step 4: Conclude. }  From \eqref{eq:IR_NE}, \eqref{eq:IR_expand}, and \eqref{eq:IR_suffices}:
\begin{equation}
    u^{\text{Shap}}_i(\vec{s}^*) \geq 
    \varphi^A_i(s^{\text{solo}}_i, \vec{s}^*_{-i}) 
    - \gamma s^{\text{solo}}_i \geq 
    1-(s^{\text{solo}}_i+1)^{-\beta} - \gamma 
    s^{\text{solo}}_i = v^*_i.
\end{equation}

Therefore $\Delta_i(\mathcal{M}^{\text{Shap}}) = u^{\text{Shap}}_i(\vec{s}^*) - v^*_i \geq 0$
for all $i \in N$. \qed
    
\end{proof}

\subsection{Proof of Theorem \ref{thm:IR_ext}}
\label{app:IR_ext}

By the game form described in Section 3.2, every Nash equilibrium of $\mathcal{M}^E$ implements a social optimum $\vec{s}^*$, and under Assumption \ref{assump:interior}, the equilibrium payment coincides with \eqref{eq:ext_payment_explicit}. It therefore suffices to show $u^E_i(\vec{s}^*) \geq v^*_i$.

\medskip
\noindent\textbf{Step 1: Simplification of $u_i^E(\vec{s}^*)$. } Substituting the payment rule \eqref{eq:ext_payment_explicit} and cost 
$c_i(s^*_i) = \gamma_i s^*_i$ into $u_i^E(\vec{s}^*) = a_i(\vec{s}^*) 
- c_i(s^*_i) + p_i^E(\vec{s}^*)$:
\begin{align}
    u_i^E(\vec{s}^*) 
    &= a_i(\vec{s}^*) - \gamma_i s^*_i + \gamma_i s^*_i 
    - \beta \xi_i(\vec{s}^*)^{-\beta-1}\left(\xi_i(\vec{s}^*) - 1\right) 
    \notag \\
    &= a_i(\vec{s}^*) - \beta \xi_i(\vec{s}^*)^{-\beta-1}
    \left(\xi_i(\vec{s}^*) - 1\right),
\end{align}
where the $\gamma_i s^*_i$ terms cancel. Substituting $a_i(\vec{s}^*) = 
1 - \xi_i(\vec{s}^*)^{-\beta}$ from \eqref{eq:payoff} and writing $\xi_i 
:= \xi_i(\vec{s}^*)$ for brevity:
\begin{align}
    u_i^E(\vec{s}^*) 
    &= 1 - \xi_i^{-\beta} - \beta\xi_i^{-\beta-1}(\xi_i - 1) \notag \\
    &= 1 - \xi_i^{-\beta} - \beta\xi_i^{-\beta} + \beta\xi_i^{-\beta-1} 
    \notag \\
    &= 1 - (1+\beta)\xi_i^{-\beta} + \beta\xi_i^{-\beta-1},
    \label{eq:u_ext_simplified}
\end{align}
where we expanded $\beta\xi_i^{-\beta-1}(\xi_i - 1) = \beta\xi_i^{-\beta} 
- \beta\xi_i^{-\beta-1}$.

\medskip
\noindent\textbf{Step 2: Simplification of $v_i^*$. } Recall the standalone utility from \eqref{eq:standalone}. When all other 
agents contribute zero, $\xi_i(s_i, \mathbf{0}_{-i}) = s_i + 1$ since 
$w_{ii} = 1$ and $w_{ij} s_j = 0$ for $j \neq i$. So:
\begin{equation}
    v_i^* = \max_{s_i \in S_i} \left[ 1 - (s_i + 1)^{-\beta} 
    - \gamma_i s_i \right].
\end{equation}

Under Assumption \ref{assump:interior_solo}, $s_i^{\text{solo}} \in (0, \tau_i)$, 
so the first-order condition holds as an equality:
\begin{equation}
    \beta(s_i^{\text{solo}} + 1)^{-\beta-1} = \gamma_i.
    \label{eq:solo_foc}
\end{equation}

Note that \eqref{eq:solo_foc} requires $\beta > \gamma_i$, since the 
marginal benefit at zero standalone contribution is 
$\beta(0+1)^{-\beta-1} = \beta$, which must exceed $\gamma_i$ for the 
lower bound $s_i^{\text{solo}} = 0$ to not be active; this is necessary 
and sufficient for $s_i^{\text{solo}} > 0$ and is guaranteed by Assumption \ref{assump:interior_solo}. Substituting \eqref{eq:solo_foc} into 
$v_i^* = 1 - (s_i^{\text{solo}}+1)^{-\beta} - \gamma_i s_i^{\text{solo}}$:
\begin{align}
    v_i^* 
    &= 1 - (s_i^{\text{solo}}+1)^{-\beta} 
    - \beta(s_i^{\text{solo}}+1)^{-\beta-1} s_i^{\text{solo}} \notag \\
    &= 1 - (s_i^{\text{solo}}+1)^{-\beta} - \beta(s_i^{\text{solo}}+1)^{-\beta} 
    + \beta(s_i^{\text{solo}}+1)^{-\beta-1} \notag \\
    &= 1 - (1+\beta)(s_i^{\text{solo}}+1)^{-\beta} 
    + \beta(s_i^{\text{solo}}+1)^{-\beta-1},
    \label{eq:v_star_simplified}
\end{align}
where we expanded $\beta(s_i^{\text{solo}}+1)^{-\beta-1} s_i^{\text{solo}} 
= \beta(s_i^{\text{solo}}+1)^{-\beta} - \beta(s_i^{\text{solo}}+1)^{-\beta-1}$. Comparing \eqref{eq:u_ext_simplified} and \eqref{eq:v_star_simplified}, 
both $u_i^E(\vec{s}^*)$ and $v_i^*$ have the same functional form, 
evaluated at $\xi_i(\vec{s}^*)$ and $s_i^{\text{solo}} + 1$ respectively. 

\medskip
\noindent\textbf{Step 3: Define $f$ and show monotonicity. }
\begin{lemma}
\label{lem:f_monotone}
The function $f: [1,\infty) \to \mathbb{R}$ defined by:
\begin{equation}
    f(x) := (1+\beta)x^{-\beta} - \beta x^{-\beta-1}
    \label{eq:f_def}
\end{equation}
is non-increasing on $[1, \infty)$.
\end{lemma}

\begin{proof}
Taking the derivative:
\begin{equation}
    f'(x) = -\beta(1+\beta)x^{-\beta-1} + \beta(\beta+1)x^{-\beta-2} 
    = \beta(\beta+1)x^{-\beta-2}(1-x) \leq 0,
\end{equation}
for all $x \geq 1$, since $\beta > 0$ and $(1-x) \leq 0$.
\end{proof}

\medskip
\noindent\textbf{Step 4: Federation provides more data. }
\begin{lemma}
\label{lem:xi_geq_solo}
Under Assumptions \ref{assump:interior} and \ref{assump:interior_solo} and $w_{ij} \geq 0$ 
for all $i, j \in N$:
\begin{equation}
    \xi_i(\vec{s}^*) \geq s_i^{\text{solo}} + 1 \quad \forall i \in N.
\end{equation}
\end{lemma}

\begin{proof}
From the social optimum KKT condition \eqref{eq:KKT_interior} and the 
standalone FOC \eqref{eq:solo_foc}:
\begin{equation}
    \beta \sum_{k \in N} w_{ki} \xi_k(\vec{s}^*)^{-\beta-1} = \gamma_i 
    = \beta(s_i^{\text{solo}}+1)^{-\beta-1}.
    \label{eq:equal_gamma}
\end{equation}
Since $w_{ii} = 1$ and $w_{ki} \geq 0$ for all $k$:
\begin{equation}
    \beta \xi_i(\vec{s}^*)^{-\beta-1} \leq \beta \sum_{k \in N} w_{ki} 
    \xi_k(\vec{s}^*)^{-\beta-1} = \beta(s_i^{\text{solo}}+1)^{-\beta-1},
\end{equation}
where the inequality holds because the sum includes the $k=i$ term 
$\beta \xi_i(\vec{s}^*)^{-\beta-1}$ plus additional non-negative terms. 
Dividing by $\beta$ and since $x \mapsto x^{-\beta-1}$ is strictly 
decreasing:
\begin{equation}
    \xi_i(\vec{s}^*)^{-\beta-1} \leq (s_i^{\text{solo}}+1)^{-\beta-1} 
    \implies \xi_i(\vec{s}^*) \geq s_i^{\text{solo}} + 1.
\end{equation}
\end{proof}

\medskip
\noindent\textbf{Step 5: Conclude IR. }
From \eqref{eq:u_ext_simplified} and \eqref{eq:v_star_simplified}, we 
can write both utilities in terms of the function $f$ defined in 
\eqref{eq:f_def}:
\begin{equation}
    u_i^E(\vec{s}^*) = 1 - f(\xi_i(\vec{s}^*)), \qquad 
    v_i^* = 1 - f(s_i^{\text{solo}} + 1).
\end{equation}
Therefore:
\begin{equation}
    u_i^E(\vec{s}^*) - v_i^* = f(s_i^{\text{solo}} + 1) - 
    f(\xi_i(\vec{s}^*)).
    \label{eq:IR_gap_app}
\end{equation}
By \Cref{lem:xi_geq_solo}, $\xi_i(\vec{s}^*) \geq s_i^{\text{solo}} + 1 
\geq 1$, where the second inequality holds since $s_i^{\text{solo}} \geq 
0$. By \Cref{lem:f_monotone}, $f$ is non-increasing on $[1, \infty)$, 
so:
\begin{equation}
    f(\xi_i(\vec{s}^*)) \leq f(s_i^{\text{solo}} + 1).
\end{equation}
Substituting into \eqref{eq:IR_gap_app}:
\begin{equation}
    u_i^E(\vec{s}^*) - v_i^* = f(s_i^{\text{solo}} + 1) - 
    f(\xi_i(\vec{s}^*)) \geq 0,
\end{equation}
establishing $u_i^E(\vec{s}^*) \geq v_i^*$ for all $i \in N$. \qed

\section{Supplementaries of Section 6}
We begin by defining and proving Lemma \ref{lemma:symmetric_NE}, which is then applied to \Cref{thm:reciprocity_ME_homogeneous}.

\begin{lemma}
\label{lemma:symmetric_NE}
    In the homogeneous case ($w_{ij} = 1$ for all $i,j \in N$, $\gamma_i = \gamma$, and $\tau_i = \tau$ for all $i \in N$), under Assumptions \ref{assump:interior} and \ref{assump:interior_solo}, there exists a symmetric Nash equilibrium of $\mathcal{M}^E$ at which $s^*_i = \bar{s}$ for all $i \in N$, for some $\bar{s} > 0$.
\end{lemma}

\subsection{Proof of Lemma \ref{lemma:symmetric_NE}}
\label{app:symmetric_NE}

\noindent\textbf{Step 1: Reduction of welfare to a scalar problem. } In the homogeneous case $w_{ij} = 1$ for all $i,j \in N$, the weighted contribution received by agent $i$ simplifies to:
\begin{equation}
    \xi_i(\vec{s}) = \sum_{j \in N} w_{ij} s_j + 1 = \|\vec{s}\|_1 + 1,
    \label{eq:xi_homogeneous}
\end{equation}

which is identical for every agent $i \in N$ and depends on $\vec{s}$ only through the scalar $\|\vec{s}\|_1$. Substituting 
$a_i(\vec{s}) = 1 - \xi_i(\vec{s})^{-\beta}$ from \eqref{eq:payoff} 
into the welfare function \eqref{eq:welfare}:
\begin{align}
    W(\vec{s}) 
    &= \sum_{i \in N} \left[a_i(\vec{s}) - \gamma s_i\right] \notag \\
    &= \sum_{i \in N} \left[1 - (\|\vec{s}\|_1 + 1)^{-\beta} 
    - \gamma s_i\right] \notag \\
    &= n\left[1 - (\|\vec{s}\|_1+1)^{-\beta}\right] 
    - \gamma\|\vec{s}\|_1.
\end{align}

Therefore $W(\vec{s})$ depends on $\vec{s}$ only through the scalar $z := \|\vec{s}\|_1$. Define:
\begin{equation}
    h(z) := n\left[1-(z+1)^{-\beta}\right] - \gamma z, 
    \quad z \in [0, n\tau],
    \label{eq:h_def}
\end{equation}
so that $W(\vec{s}) = h(\|\vec{s}\|_1)$ for all $\vec{s} \in \mathcal{S}$. Therefore maximizing $W$ over $\mathcal{S}$ is equivalent to finding $\vec{s} \in \mathcal{S}$ with $\|\vec{s}\|_1 = z^*$, where $z^*$ is the maximizer of $h$ over $[0, n\tau]$.

\medskip
\noindent\textbf{Step 2: Unique interior maximizer of $h$. } Computing the first and second derivatives of $h$:
\begin{equation}
    h'(z) = n\beta(z+1)^{-\beta-1} - \gamma, \qquad
    h''(z) = -n\beta(\beta+1)(z+1)^{-\beta-2}.
\end{equation}
Since $\beta > 0$, $n \geq 1$, and $(z+1)^{-\beta-2} > 0$ for all $z \geq 0$, we have $h''(z) < 0$ for all $z \geq 0$, so $h$ is  strictly concave on $[0, n\tau]$ and has at most one maximizer. We show the maximizer $z^*$ is interior to $(0, n\tau)$:

\begin{enumerate}
    \item By Assumption \ref{assump:interior_solo}, $s_i^{\text{solo}} > 0$, so the standalone FOC gives $\gamma = \beta(s_i^{\text{solo}}+1)^{-\beta-1} < \beta$, 
    since $(s_i^{\text{solo}}+1)^{-\beta-1} < 1$ for $s_i^{\text{solo}} > 0$. Therefore $h'(0) = n\beta - \gamma > 0$, so $h$ is increasing at $z = 0$ and $z^* > 0$.

    \item $z^*$ exists in $(0, \infty)$: Since $h'(0) > 0$ and $h'(z) \to -\gamma < 0$ as $z \to \infty$, by the intermediate value theorem there exists $z^* \in (0, \infty)$ with $h'(z^*) = 0$. By strict concavity this $z^*$ is unique. 

    \item $z^* < n\tau$: Assumption \ref{assump:interior} presupposes that $\mathcal{M}^E$ admits an equilibrium whose implemented profile $\vec{s}'$ is interior, i.e. $s_i' < \tau$ for all $i \in N$. Since $\vec{s}'$ is a welfare maximizer and $W = h(\|\cdot\|_1)$, every welfare maximizer has total contribution $z^*$, so $z^* = \|\vec{s}'\|_1 < n\tau$.

\end{enumerate}

\medskip
\noindent\textbf{Step 3: Symmetric profile is a social optimum.} Define the symmetric profile $\vec{s}^*$ by $s^*_i = \bar{s} := z^*/n$ for all $i \in N$.We verify $\vec{s}^* \in \mathcal{S}$: since $z^* \in (0, n\tau)$, we have $\bar{s} = z^*/n \in (0, \tau)$, so $s^*_i \in [0, \tau]$ for all $i$. Furthermore:
\begin{equation}
    \|\vec{s}^*\|_1 = \sum_{i \in N} \bar{s} = n \cdot \frac{z^*}{n} 
    = z^*.
\end{equation}
Therefore $W(\vec{s}^*) = h(\|\vec{s}^*\|_1) = h(z^*)$, which is the maximum of $h$ over $[0, n\tau]$. Hence $\vec{s}^*$ is a social optimum of \eqref{eq:SO_def}.

\medskip
\noindent\textbf{Step 4: Existence of a symmetric Nash equilibrium.}
By Theorem 2 of \cite{sharma2012localpublicgoodprovisioning}, for every social optimum of 
\eqref{eq:SO_def} there exists at least one NE of $\mathcal{M}^E$ implementing it. Applying this to the symmetric 
social optimum $\vec{s}^*$ constructed in Step 3, there exists a NE of $\mathcal{M}^E$ at which $s^*_i = \bar{s}$ for all $i \in N$. \qed

\subsection{Proof of Theorem \ref{thm:reciprocity_ME_homogeneous}}

By \Cref{lemma:symmetric_NE}, there exists a symmetric NE of $\mathcal{M}^E$ at which $s^*_i = \bar{s}$ for 
all $i \in N$. Let $\vec{s}^*$ be any such symmetric NE. We show that $a_i(\vec{s}^*) + p^E_i(\vec{s}^*) = \varphi^A_i(\vec{s}^*)$ for all $i \in N$, which gives $r(\vec{s}^*) = 1$.

\medskip
\noindent\textbf{Step 1: Shapley value at a symmetric profile.} Since $\vec{s}^*$ is symmetric with $s^*_i = \bar{s}$ for all $i \in N$, and all agents have identical payoff functions in the homogeneous case, we show that all agents are interchangeable in the sense of the Shapley value. Specifically, for any two agents $i, j \in N$ and any coalition $X \subseteq N \setminus \{i,j\}$:

\begin{align}
    A(\vec{s}^*[X \cup \{i\}]) 
    &= n\left[1 - \left(\sum_{k \in X \cup \{i\}} \bar{s} 
    + 1\right)^{-\beta}\right] \notag \\
    &= n\left[1 - \left((|X|+1)\bar{s}+1\right)^{-\beta}\right] 
    \notag \\
    &= n\left[1 - \left(\sum_{k \in X \cup \{j\}} \bar{s} 
    + 1\right)^{-\beta}\right] \notag \\
    &= A(\vec{s}^*[X \cup \{j\}]),
\end{align}

where we used $w_{ij} = 1$ for all $i,j$ so that $\xi_k(\vec{s}^*[X \cup \{i\}]) = (|X|+1)\bar{s} + 1$ for all $k \in N$, depending only on $|X|$ and not on which specific agent is added. Therefore all agents make identical marginal contributions to every coalition, so by the symmetry of the Shapley value:
\begin{equation}
    \varphi^A_i(\vec{s}^*) = \varphi^A_j(\vec{s}^*) \quad 
    \forall i,j \in N.
    \label{eq:shapley_symmetric}
\end{equation}

By the efficiency property of the Shapley value:
\begin{equation}
    \sum_{i \in N} \varphi^A_i(\vec{s}^*) = A(\vec{s}^*) = 
    \sum_{i \in N} a_i(\vec{s}^*) = n \cdot a(\vec{s}^*),
\end{equation}
where the last equality uses $a_i(\vec{s}^*) = a(\vec{s}^*)$ for all $i$ at the symmetric profile. Combining with \eqref{eq:shapley_symmetric}:

\begin{equation}
    \varphi^A_i(\vec{s}^*) = \frac{n \cdot a(\vec{s}^*)}{n} = 
    a(\vec{s}^*) \quad \forall i \in N.
    \label{eq:shapley_equals_payoff}
\end{equation}

\medskip
\noindent\textbf{Step 2: Payment vanishes at symmetric Nash equilibrium.} In the homogeneous case with $s^*_i = \bar{s}$ for all $i$, the weighted contribution received by agent $i$ is:

\begin{equation}
    \xi_i(\vec{s}^*) = \sum_{j \in N} w_{ij} s^*_j + 1 = 
    n\bar{s} + 1 =: \xi,
    \label{eq:xi_symmetric}
\end{equation}

which is identical for every agent $i \in N$. Since $\vec{s}^*$ is a Nash equilibrium of $\mathcal{M}^E$, it implements a welfare maximizer, and since its allocation is interior ($\bar{s} \in (0, \tau)$ by Lemma \ref{lemma:symmetric_NE}), it satisfies the KKT condition \eqref{eq:KKT_interior}. In the homogeneous case, substituting $w_{ki} = 1$ and $\xi_k(\vec{s}^*) = \xi$ for all $k \in N$ into \eqref{eq:KKT_interior}:

\begin{equation}
    \beta \sum_{k \in N} w_{ki} \xi_k(\vec{s}^*)^{-\beta-1} = 
    \beta n \xi^{-\beta-1} = \gamma.
    \label{eq:KKT_symmetric}
\end{equation}

Therefore $\beta\xi^{-\beta-1} = \gamma/n$. Substituting $s^*_i = \bar{s}$, $\xi_i(\vec{s}^*) = \xi$, and $\beta\xi^{-\beta-1} = \gamma/n$ into the payment rule \eqref{eq:ext_payment_explicit}:

\begin{align}
    p^E_i(\vec{s}^*) 
    &= \gamma\bar{s} - \beta\xi^{-\beta-1}(\xi - 1) \notag \\
    &= \gamma\bar{s} - \frac{\gamma}{n}(n\bar{s} + 1 - 1) 
    \notag \\
    &= \gamma\bar{s} - \frac{\gamma}{n} \cdot n\bar{s} \notag \\
    &= \gamma\bar{s} - \gamma\bar{s} = 0.
    \label{eq:payment_vanishes}
\end{align}

\medskip
\noindent\textbf{Step 3: Conclude reciprocity.} From \eqref{eq:shapley_equals_payoff} and \eqref{eq:payment_vanishes}, for every agent $i \in N$:
\begin{equation}
    a_i(\vec{s}^*) + p^E_i(\vec{s}^*) = a(\vec{s}^*) + 0 = 
    a(\vec{s}^*) = \varphi^A_i(\vec{s}^*).
\end{equation}
Therefore:
\begin{equation}
    r(\vec{s}^*) = \min_{i \in N} \frac{a_i(\vec{s}^*) +
    p^E_i(\vec{s}^*)}{\varphi^A_i(\vec{s}^*)} = 1.
\end{equation}
Since $\vec{s}^*$ was an arbitrary symmetric Nash equilibrium, this holds at every symmetric Nash equilibrium of $\mathcal{M}^E$. \qed

\section{Supplementaries of Section 8}

\subsection{Proof of Proposition \ref{prop:Ti}}
\label{app:prop_Ti}

We allow $w_{ij} \in \mathbb{R}$ for all $i, j \in N$, with $w_{ii} = 1$, and suppose the domain condition $\xi_k(\vec{s}) > 0$ holds for all $\vec{s} \in S$ and all $k \in N$. We fix agent $i \in N$ and assume $\xi_i(\vec{s}^*) \geq 1$. Under Assumptions~\ref{assump:interior} and~\ref{assump:interior_solo}, we show $T_i(\vec{s}^*) \geq 0 \iff \Delta_i(\mathcal{M}^E) \geq 0$.

\medskip
\noindent\textbf{Step 1: Simplification of $u_i^E(\vec{s}^*)$. } The derivation of $u_i^E(\vec{s}^*)$ does not depend on the signs of  $w_{ij}$, only on the payment rule \eqref{eq:ext_payment_explicit} and the payoff function \eqref{eq:payoff}. Therefore, by the same calculation as in Step 4--5 of the proof of Theorem~\ref{thm:IR_ext}:

\begin{equation}
    u_i^E(\vec{s}^*) = 1 - f(\xi_i(\vec{s}^*)), 
    \qquad v_i^* = 1 - f(s_i^{\text{solo}} + 1),
    \label{eq:prop_uv}
\end{equation}
where $f$ is defined in \eqref{eq:f_def}, so that:
\begin{equation}
    \Delta_i(\mathcal{M}^E) = u_i^E(\vec{s}^*) - v_i^* 
    = f(s_i^{\text{solo}} + 1) - f(\xi_i(\vec{s}^*)).
    \label{eq:prop_gap}
\end{equation}

\medskip
\noindent\textbf{Step 2: Relating $T_i(\vec{s}^*)$ to $\xi_i(\vec{s}^*)$ and $s_i^{\text{solo}} + 1$. } Since $\vec{s}^*$ implements an interior social optimum under Assumption~\ref{assump:interior}, the KKT stationarity condition \eqref{eq:KKT_interior} holds with equality. Isolating the $k = i$ term using $w_{ii} = 1$:

\begin{equation}
    \beta \xi_i(\vec{s}^*)^{-\beta-1} + \beta \sum_{k \neq i} 
    w_{ki} \xi_k(\vec{s}^*)^{-\beta-1} = \gamma_i,
\end{equation}
which, using the definition of $T_i(\vec{s}^*)$ from \eqref{eq:Ti}, gives:

\begin{equation}
    \beta \xi_i(\vec{s}^*)^{-\beta-1} = \gamma_i - \beta T_i(\vec{s}^*).
    \label{eq:prop_kkt}
\end{equation}

Under Assumption~\ref{assump:interior_solo}, the standalone FOC \eqref{eq:solo_foc} gives $\gamma_i = \beta(s_i^{\text{solo}}+1)^{-\beta-1}$. Substituting into \eqref{eq:prop_kkt} and dividing by $\beta$:

\begin{equation}
    \xi_i(\vec{s}^*)^{-\beta-1} - (s_i^{\text{solo}}+1)^{-\beta-1} 
    = -T_i(\vec{s}^*).
    \label{eq:prop_star}
\end{equation}

Since $x \mapsto x^{-\beta-1}$ is strictly decreasing on $(0, \infty)$, and both $\xi_i(\vec{s}^*) > 0$ (by the domain condition) and $s_i^{\text{solo}} + 1 > 0$:

\begin{equation}
    T_i(\vec{s}^*) \gtrless 0 
    \iff \xi_i(\vec{s}^*)^{-\beta-1} \lessgtr 
    (s_i^{\text{solo}}+1)^{-\beta-1}
    \iff \xi_i(\vec{s}^*) \gtrless s_i^{\text{solo}} + 1.
    \label{eq:prop_xi_compare}
\end{equation}

\medskip
\noindent\textbf{Step 3: Conclude. } By hypothesis, $\xi_i(\vec{s}^*) \geq 1$, and by Assumption~\ref{assump:interior_solo}, $s_i^{\text{solo}} > 0$ so $s_i^{\text{solo}} + 1 > 1$. Thus both arguments of $f$ in \eqref{eq:prop_gap} lie in $[1, \infty)$. Since $f'(x) = \beta(\beta+1)x^{-\beta-2}(1-x) \leq 0$ for all $x \geq 1$ with equality only at $x = 1$, $f$ is strictly decreasing on $[1, \infty)$.Therefore:

\begin{equation}
    \xi_i(\vec{s}^*) \gtrless s_i^{\text{solo}} + 1 
    \iff f(\xi_i(\vec{s}^*)) \lessgtr f(s_i^{\text{solo}} + 1) 
    \iff \Delta_i(\mathcal{M}^E) \gtrless 0.
    \label{eq:prop_f_compare}
\end{equation}

Chaining \eqref{eq:prop_xi_compare} and \eqref{eq:prop_f_compare}:

\begin{equation}
    T_i(\vec{s}^*) \gtrless 0 \iff \Delta_i(\mathcal{M}^E) \gtrless 0,
\end{equation}

which gives the stated result $T_i(\vec{s}^*) \geq 0 \iff \Delta_i(\mathcal{M}^E) \geq 0$. \qed

\section{MNIST Simulation Setup}
\label{app:sim}
Three agents each hold 400 MNIST images (10 classes, 20\% test split, initial contribution $s_i^{(0)} = 150$) and train a three-layer MLP ($784 \to 128 \to 64 \to 10$) with ReLU and log-softmax via SGD (lr $= 0.05$, momentum $= 0.9$, 1 local epoch) over 50 rounds. For $\mathcal{M}^{Shap}$, best-response dynamics use step size $10^3$, gradient perturbation $\varepsilon = 10$, and a stability threshold of $5.0$ over 3 rounds, with Shapley values estimated via 300 Monte Carlo permutations. The $\mathcal{M}^E$ equilibrium is found by grid search over $\pm 100$ samples in steps of 20, training each candidate from scratch for 20 epochs.
\end{document}